\documentclass[a4paper,twocolumn,11pt, accepted= 2021-09-26]{quantumarticle}
\pdfoutput=1
\usepackage{amsmath,amsthm,amsfonts,amssymb}
\usepackage{graphicx}
\usepackage{hyperref}
\usepackage[capitalize,nameinlink]{cleveref}
\usepackage{color}
\usepackage{braket}
\numberwithin{equation}{section}
\numberwithin{figure}{section}
\usepackage[numbers,sort&compress]{natbib}

\newtheorem{problem}{Problem}
\newtheorem{theorem}{Theorem}
\newtheorem{lemma}{Lemma}

\newtheorem{definition}{Definition}

\newtheorem{corollary}{Corollary}
\theoremstyle{remark}
\newtheorem{remark}{Remark}
\DeclareMathAlphabet{\mathpzc}{OT1}{pzc}{m}{it}

\newcommand{\be}{\begin{equation}}
\newcommand{\ee}{\end{equation}}
\newcommand{\Sw}{{\rm Swap}}

\newcommand{\poly}{\operatorname{poly}}

\begin{document}

\title{Symmetry Protected Quantum Computation}

\author{M. H. Freedman}
\affiliation{Station Q, Microsoft Research, Santa Barbara, CA 93106-6105, USA}
\affiliation{Microsoft Quantum, Redmond, WA 98052, USA}

\author{M.~B.~Hastings}
\affiliation{Station Q, Microsoft Research, Santa Barbara, CA 93106-6105, USA}
\affiliation{Microsoft Quantum, Redmond, WA 98052, USA}

\author{M. Shokrian Zini}
\affiliation{Perimeter Institute for Theoretical Physics, Waterloo, ON N2L 2Y5, Canada}
\affiliation{Research Consultant, Microsoft}

\begin{abstract}
We consider a model of quantum computation using qubits where it is possible to measure whether a given pair are in a singlet (total spin $0$) or triplet (total spin $1$) state.
The physical motivation is that we can do these measurements in a way that is protected against revealing other information so long as all terms in the Hamiltonian are $SU(2)$-invariant.
We conjecture that this model is equivalent to BQP.  Towards this goal, we show:
(1) this model is capable of universal quantum computation with polylogarithmic overhead if it is supplemented by single qubit $X$ and $Z$ gates.
(2) Without any additional gates, it is at least as powerful as the weak model of ``permutational quantum computation" of Jordan \cite{M05,jordan2010permutational}.
(3) With postselection, the model is equivalent to PostBQP.
\end{abstract}
\maketitle

Imperfect physical gates are a major challenge for building a scalable quantum computer.  One possible way to overcome this challenge is to use error correcting codes to build high fidelity logical gates from the lower fidelity physical gates \cite{Gottesman_2010}.  Another approach is to use a topologically ordered state to store and manipulate quantum information, directly obtaining good logical gates \cite{Kitaev_2003}.  Here, we propose a third approach, to protect the operations by symmetries of the physical Hamiltonian.

In particular, we consider qubits encoded in quantum spins, and we assume that the Hamiltonian and any noise terms respect an $SU(2)$ symmetry acting on all the qubits simultaneously. We need a quick introduction into the representation theory of $SU(2)$. The irreducible representations of $SU(2)$ are indexed by a quantity $S\in\{0,1/2,1,3/2,\ldots\}$, called the spin.  The dimensionality of a representation of spin $S$ is $2S+1$.  The spin-$1/2$ representation has dimension $2$, so can be regarded as a qubit.  The tensor product of two spin-$1/2$ representations is the direct sum of a spin $0$ and spin $1$ representation. The singlet state of a pair of qubits is, up to an arbitrary phase, $\frac{1}{\sqrt{2}}(|01\rangle-|10\rangle)$, while states $\{|00\rangle, |11\rangle,\frac{1}{\sqrt{2}}(|01\rangle+|10\rangle)\}$ in the subspace orthogonal to the singlet state are called triplet.  The singlet state is anti-symmetric under exchange of qubits, while the triplet states are symmetric.

In the idealized case, a single qubit is not subject to any noise of any kind, since there are no terms that we can write down which are invariant under $SU(2)$ on that qubit. If we bring a pair of qubits together,  then the total spin (either $0$ or $1$) is the only $SU(2)$-invariant term on both qubits.  Thus, dephasing may happen between spin $0$ (singlet) and spin $1$ (triplet) states; however, it also becomes possible to {\it measure} this total spin.

Our focus in this paper is on the power of these singlet/triplet measurements in an idealized model.
To state our physical model: we assume that there are $n$ qubits for some $n$, initially in some state which is a tensor product of singlets.  We then allow arbitrary pairs of qubits to be selected and the total spin of those two qubits to be selected.
We use the notation $s/t$ to indicate this operation, which projectively measures a pair of qubits to be in either the singlet ($s$) or triplet ($t$) subspace.
From this physical model, we define a computational complexity class that we call STP, where a sequence of polynomially many such $s/t$ measurements is allowed. More precisely,
\begin{definition}
STP is the class of problems that can be correctly answered, with constant probability larger than $\frac{1}{2}$, using polynomially many $s/t$ measurements where the sequence of measurements to make is determined by a polynomial time classical algorithm based on the input string and on the previous measurement outcomes.  Similarly, the decision of whether the input string is accepted or rejected is also determined by a polynomial time classical algorithm based on the string and the measurement outcomes.
\end{definition}

As perhaps the ``smallest" physical example where this can be realized, imagine a single electron and a single proton far separated from each other.  Each has spin $1/2$, and (assuming no magnetic fields are present and assuming that spin-orbit coupling can be ignored which depends on the electric fields and on the orbital degrees of freedom) then these spins are not subject to any decoherence.  If the electron and proton are brought near each other, then there is a hyperfine coupling between these spins, and the ground state splits into singlet and triplet levels with slightly different energies which can in principle be measured.  One might then try to compute by having a large number of electrons and protons and occasionally bringing a single electron near a single proton and measuring the total spin.  This simple example is slightly different from the idealized model above, as some of the qubits are associated with electrons and some with protons, and pairs must have one electron and one proton.

More realistic examples include trapped ions where interaction between ions may be dependent on total spin, or spin blockade in quantum dots \cite{Johnson_2005}.

This paper considers the question: what is the computational power of STP? We believe the study of such a simple model, in addition to its symmetry protection, is of interest itself and we conjecture that STP is equivalent to BQP. Even if our conjecture is wrong and STP falls somewhere between P and BQP, the study of STP is still well-motivated as it is so distinct from other quantum computational models. While we do not show STP = BQP, we present several partial results\footnote{After this paper appeared, it was shown that STP=BQP \cite{rudolph2021relational}, building on \cite{Rudolph_2005}. Even given that STP=BQP, the constructions here are of interest as they may give a lower overhead way to implement BQP in practice using additional approximate gates.}.

We emphasize that this kind of physical protection by symmetry is distinct from the approach using
measurement based quantum computation in a symmetry-protected topological phase \cite{Nautrup_2015}.

First, in \cref{secxzst}, we show how to implement universality using $s/t$ measurements as well as certain single qubit Cliffords and Pauli measurements.  Our discussion will be in reverse order of the number of single qubit operations we allow: we will start with the full single qubit Clifford group and all Pauli measurements, and gradually reduce the number of single qubit operations required, until what is required is
$X,Z$ Clifford operations in addition to $s/t$ measurements.  A related previous result is that $s/t$ operations combined with a source of many copies of three different states whose Bloch vectors span is also universal \cite{Rudolph_2005}; \textit{s/t} operations are used to purify mixed states and build any pure state, which are then used to build cluster states. This purification result, impossible to do using Clifford operations, further motivates the study of STP. We refer to the introduction in \cite{Rudolph_2005} for further motivation behind the study of STP.

In \cref{secpqc}, we show that STP is at least as powerful as the model of ``weak permutational computing" \cite{jordan2010permutational} (it has been shown \cite{H18} that there is an efficient classical algorithm to compute amplitudes in this permutational computing model, but no classical algorithm is known to solve the sampling problem in this model).

In \cref{secgtm}, in line with our investigation of STP, we consider generalizations of STP, showing that allowing higher spin qudits does not increase the power of the model.

Finally, in \cref{postbqp}, we define a postselected version of STP, called PostSTP, and show that it is equivalent to PostBQP. This implies that efficient \textit{exact} sampling from an STP protocol is impossible assuming the polynomial hierarchy is infinite \cite[See ``Sidebar: The Polynomial Hierarchy and Post-Selection'']{harrow2017quantum}. An approximate efficient sampling, related to quantum supremacy experiments, is discussed in \cref{statsofampsappendix} (\cref{STSampleproblem}). Showing PostSTP = PostBQP does not imply STP = BQP, as there are other models, such as boson sampling and IQP which are not believed to be BQP and yet their postselection version equals PostBQP (see \cite{harrow2017quantum} for more on PostBQP and quantum supremacy experiments).

There are several interesting open questions, such as the obvious question of the hardness of sampling measurement outcomes for some sequence of $s/t$ measurements, with the sequence randomly chosen or chosen in some other way.

A final result that may be of independent interest is given in \cref{irrmeasure}, namely a relation between the ``irrationality measure" of an irrational angle and how many times one may need to rotate by that angle to approximate some desired target angle to some desired accuracy. Although we use this result in \cref{uc}, it is only to show an efficient approximation that can also be done with more elementary facts. Likely the result in \cref{irrmeasure} is well-known, but since we did not find it elsewhere we record it here.
\section{$X,Z$ and $s/t$ is Universal for Quantum Computation}
\label{secxzst}
We show that:
\begin{theorem}
\label{lemxzst}
Using $s/t$ measurements, and single qubit $X,Z$ unitaries,
we can approximate a gate set consisting of one and two qubit Clifford operations and Pauli measurements, as well as $T$ gates, with an expected overhead that is only polylogarithmic in the error in this approximation.
\end{theorem}
\begin{remark}
The polylogarithmic overhead arises for a familiar reason: we will exactly implement Clifford gates and then implement $T$ gates by magic state distillation using as input approximate $T$ gates that we can form from the given operations.  The overhead scales polylogarithmically in the accuracy of the distilled $T$ gates.  The implementation of the circuit is probabilistic, as certain operations are repeated until they succeed, which is why the theorem refers to an {\it expected} overhead.
\end{remark}

\subsection{$s/t$ and Single Qubit Clifford and Pauli Measurement is Universal for Quantum Computing}
We first show that $s/t$ on arbitrary pairs of qubits, combined with arbitrary Cliffords on a single qubit (i.e., $X,Z,H,S$) and single qubit Pauli measurements, is universal for quantum computing.
In particular, we can approximate a gate set consisting of one and two qubit Cliffords operations and Pauli measurements, as well as $T$ gates, with an overhead that is only polylogarithmic in the error.

In subsequent sections, we reduce the number of single qubit operations required at the cost of making the construction more complicated.

We begin by showing that we can implement the full Clifford group, and then show universality.

\subsubsection{Implementing the Full Clifford Group}
\label{ifcg}
First, we prove:
\begin{lemma}[Bell basis measurement]
We can implement a four outcome projective measurement on two qubits $A,B$ in a Bell basis:
$$\frac{1}{\sqrt{2}}\Bigl(|01\rangle - |10\rangle\Bigr);\quad
\frac{1}{\sqrt{2}}\Bigl(|01\rangle + |10\rangle\Bigr);$$
$$\frac{1}{\sqrt{2}}\Bigl(|00\rangle + |11\rangle\Bigr);\quad
\frac{1}{\sqrt{2}}\Bigl(|00\rangle - |11\rangle\Bigr).$$
\end{lemma}
\begin{proof}
To do this, first apply an $s/t$ measurement.  If the outcome is $s$, then the qubits are in the first Bell state.  If the outcome is $t$, apply $Z_A$ and again measure $s/t$.  If the outcome is $s$, then the qubits were in the second Bell state.  If the outcome is $t$, apply $X_A$ and again measure $s/t$; if the outcome is $s$, then the qubits were in the third Bell state.  If the outcome is $t$, then the qubits were in the fourth Bell state.
\end{proof}
\begin{remark}
Since we can do Bell basis measurements, we can teleport: given a pair of qubits labeled $B,C$ in a Bell state, we can bring in an additional qubit labeled $A$ and measure $A,B$ in the Bell basis.  This teleports the state of $A$ to $C$, up to some correction on $C$ which is a single qubit Clifford.
\end{remark}
Further, we can teleport more than one qubit, and use this to perform Clifford operations if we can prepare appropriately entangled states. To prove that for the case of CNOT, consider a state $\psi_{CNOT}$ on four qubits $C,D,E,F$ obtained by taking $C,E$ in a singlet and $D,F$ in a singlet, and then applying a CNOT from $E$ to $F$.  We then bring in two extra qubits $A,B$, and measure $A,C$ in the Bell basis and measure $B,D$ in a Bell basis.  This teleports the state of $A,B$ to $E,F$ and then applies a CNOT, again up to some single qubit Cliffords on $E,F$.

So, if we can prepare this state $\psi_{CNOT}$, then we can perform CNOT operations.  Combined with the ability to perform single qubit Cliffords, this gives the full Clifford group.  It seems at this point that we are trying to get a ``free lunch": the most obvious way to prepare $\psi_{CNOT}$ is to perform a CNOT operation, which is precisely the operation we are trying to produce.

However, it is possible to prepare $\psi_{CNOT}$ using just $s/t$ and single qubit Cliffords, as we now show.  First we construct a double spin operation $O_{ZZ}$.
\begin{lemma}[$O_{ZZ}$]
We can produce an operation $O_{ZZ}$ which acts on a pair of qubits $A,B$ and projectively measures $Z_A Z_B$ and, if $Z_A Z_B=-1$, then it also measures $X_A X_B$.
\end{lemma}
\begin{proof}
To do this, simply perform only the first two measurements of the protocol above to measure in the Bell basis: if either of the first two measurements are $s$, then $Z_A Z_B=-1$ and we have also measured $X_A X_B$.  If both measurements are $t$, then $Z_A Z_B=+1$ but no other information is revealed; we should then apply $Z_A$ to undo the first application of $Z_A$.
\end{proof}
\begin{corollary}[$\hat{O}_{ZZ}$]
Using the operation $O_{ZZ}$, it is possible to produce an operation $\hat O_{ZZ}$ which always measures $Z_A Z_B$ and, with probability $1/2$, also measures $X_A X_B$.
\end{corollary}
\begin{proof}
To do this, flip an unbiased coin: if it is heads, apply $O_{ZZ}$; if it is tails, apply $X_A$, then apply $O_{ZZ}$, then apply $X_A$ again.  Then, if $Z_AZ_B=1$ we measure $X_A X_B$ if the coin is tails but not if it is heads, while if $Z_A Z_B=-1$, we measure $X_A X_B$ if the coin is heads but not if it is tails.
\end{proof}
\begin{remark}[$O_{XX},O_{YY},\hat{O}_{XX},\hat{O}_{YY}$]
Similarly, we can produce an operation $O_{XX}$ which measures $X_A X_B$ and, if $X_A X_B=-1$, also measures $Z_A Z_B$, and we can produce an operation $\hat O_{XX}$ which measures $X_A X_B$ and, with probability $1/2$ measures $Z_A Z_B$ otherwise measuring nothing else.  To construct $O_{XX}$, we use the same construction as $O_{ZZ}$ except we apply $X_A$ after the first measurement if the result is $t$, rather than applying $Z_A$.
We construct $\hat O_{XX}$ similarly.   Also (we will use this in \cref{subsecxzst}), we can produce operations $O_{YY}$ and $\hat O_{YY}$.  To construct $O_{YY},\hat O_{YY}$, we apply $Y_A=X_A Z_A$ (up to scalar phase) after the first measurement if it is $t$.
\end{remark}

We will say that $O_{ZZ}$, $\hat O_{ZZ}$ {\it succeed} if they measure only $Z_A Z_B$ without measuring additional information.  Similarly, we say that $O_{XX},\hat O_{XX}$ (respectively, $O_{YY},\hat O_{YY}$) {\it succeed} if they measure only $X_A X_B$ (respectively, $Y_A Y_B$) without revealing any additional information.
\begin{lemma}[CNOT]
We can produce a CNOT gate using an ancilla, given the ability to perform single qubit Cliffords and Pauli measurements as well as to apply $O_{ZZ}, O_{XX}$ on an arbitrary pair of qubits.
\end{lemma}
\begin{proof}
The circuit \cite{Xue_2013,Karzig_2017} is: prepare the ancilla in the $|+\rangle$ state, measure $ZZ$ on the source and ancilla, measure $XX$ on the ancilla and target, and finally measure $Z$ on the ancilla, and apply single qubit Clifford corrections if needed.  The original reference \cite{Xue_2013} writes the circuit with additional Hadamard gates so that all measurements are $Z$ or $ZZ$, but it conjugates to the circuit we give here. \newline
\indent Now, to produce the CNOT gate action, recall that we only need to prepare the state $\psi_{CNOT}$. We do so by a probabilistic protocol: prepare two entangled qubits.  Then, attempt to produce a CNOT gate by using operations $O_{ZZ}$ and $O_{XX}$ in place of the $ZZ$ and $XX$ measurements in the aforementioned protocol for a CNOT.  If both $O_{ZZ}$ and $O_{XX}$ succeed, we have produced the desired $\psi_{CNOT}$.  If one does not succeed, we may simply try again.
\end{proof}
\begin{remark}
Note that in this probabilistic protocol, indeed $O_{ZZ}$ and $O_{XX}$ will each succeed with probability $1/2$, without any need to use $\hat O_{ZZ}$ and $\hat O_{XX}$.
\end{remark}

Thus, this protocol to perform a CNOT gate can be understood simply as, first, we have a protocol that sometimes succeeds in performing the desired CNOT, and, second, by preparing entangled states ``offline". Offline operations use \textit{expendable} qubits where, without any loss in efficiency, we can start again the operations/measurements on new qubits if the desired result was not obtained. This is in contrast to \textit{online} operations applied on the \textit{data} (computational+ancilla) qubits. By teleporting through the entangled states, we can use this to perform CNOTs on data qubits by ``only using the CNOT when it will succeed".

\subsubsection{Universality}
\label{uc}
Now we show universality.  The construction here is {\it not} intended to be optimal in terms of minimizing overhead in any way.  It is simply intended to be a simple way to show universality by leveraging standard results.
\begin{lemma}
We can prepare (offline) a non-stabilizer pure state on a single qubit.
\end{lemma}
\begin{proof}
Prepare two qubits in states $|0\rangle$ and $|+\rangle$ using $Z,X$ measurements and applying $X,Z$ if the wrong outcome occurs.  Project into the triplet state (if instead they are in a singlet, re-prepare and try again).  The result is, up to normalization, $|00\rangle + \frac{1}{2}(|01\rangle +|10\rangle)$.  Measure the second qubit in the $X$ basis; we assume without loss of generality that the result is $|+\rangle$.  The result state on the first qubit is, up to normalization, $3 |0\rangle + |1\rangle$.  With normalization this is
$\cos(\theta) |0\rangle + \sin(\theta) |1\rangle$, where $\theta=\arctan(1/3)$ is an irrational angle \cite{mo}.
\end{proof}
Any such state $\cos(\theta) |0\rangle + \sin(\theta) |1\rangle$ 
is equal to $\exp(i \theta Y) |0\rangle$.  Using this as a resource for state injection produces a rotation by
$$\exp(\pm i \theta Y),$$
where the sign is chosen uniformly at random (our basis for state injection is a little different from the standard one since we rotate by $Y$ instead of by $Z$). This state injection can be done, e.g., as in \cite[Section III]{Bravyi_2005} by modifying the following appropriately using a Hadamard and phase gate:
\begin{align}\label{stateinjection}
    MZ_2 \cdot CNOT 
    (|\psi\rangle \otimes (|0\rangle + e^{i\theta} |1\rangle))  
\end{align}
$$ \to \Lambda(e^{\pm 2i\theta})|\psi\rangle $$
where $MZ_2$ is a $Z$-measurement on the second qubit, and $\Lambda(e^{\pm 2i\theta})$ the controlled $e^{\pm 2i\theta}$ rotation $\begin{pmatrix} 1 & 0 \\ 0 & e^{\pm 2i\theta} \end{pmatrix}$. Note that $XHS(\cos(\theta) |0\rangle + \sin(\theta) |1\rangle)  \propto \frac{1}{\sqrt{2}} (|0\rangle + e^{2i\theta} |1\rangle)$, with $H,S$ the Hadamard and phase gate. We shall address the issue of sign $\pm$ in $\Lambda(e^{\pm 2i\theta})$ shortly.

In the simplest applications of state injection, one imagines a situation where rotation by twice the angle is available as a primitive (for example, using state injection to produce $T$ gates and assume that one has $S$ gates available).  In that case, if the sign is not what one wants, one can recover with a rotation by twice the angle.  We do not have that option here.

We use a different approach.  Pick any desired target angle $\phi_{target}$, and any error $\epsilon>0$.  
Then, repeatedly apply state injection (like in \cref{stateinjection}) to a qubit $|\psi\rangle$ initialized in the $|0\rangle$ state, until the result is
$\cos(\phi) |0\rangle + \sin(\phi) |1\rangle$ where $\phi \equiv m\theta \pmod{2\pi}$ for some $m\in \mathbb{N}$ and $$|\phi-\phi_{target}|\leq \epsilon.$$
Here, when we apply state injection, we do not care whether the plus sign or minus sign is chosen.  The result is some random walk in angles and since $\theta$ is irrational, in expected time $O(\frac{1}{\epsilon^{\mu+1}})$, where $\mu$ is the irrationality measure of $\frac{\pi}{8}$ (\cref{irrmeasure}), we get $|\phi-\phi_{target}|\leq \epsilon$.

Choosing $\phi_{target}=\pi/8$, we can then produce states arbitrarily close to a magic state for a rotation by $\pi/8$, i.e., a magic state for a ``$T$" gate, where the quotes are because this is rotation by $Y$ rather than $Z$.  Choosing $\epsilon$ sufficiently small, one can then use this gate as input into any standard $T$ gate distillation protocol (like in \cref{stateinjection} but for $\theta \to \phi$) \cite[Section III]{Bravyi_2005} to obtain universality with only polylogarithmic overhead in the target error.

\subsection{Reducing the Single Qubit Operations Needed}
We now reduce the single qubit operations required.

\subsubsection{Avoiding Use of Hadamard and $S$ gates}
First, we note that it is possible to avoid using both Hadamard and $S$ gates.  Without these gates, our construction of \cref{ifcg} gives the subgroup of the Clifford group generated by single qubit $X$ and $Z$ and two qubit CNOT.  Call this subgroup ${\cal C}$.
We now show that we can generate the full Clifford group using just $s/t$, $X,Z$ and single qubit Pauli measurements.
\begin{lemma}
The Hadamard and phase gates can be derived by using gate distillation with the help of $CNOT$ and single qubit Pauli measurements.
\end{lemma}
\begin{proof}
Since we can measure single qubit Paulis, prepare qubits in state $Y=+1$.  Using these $Y=+1$ qubits as a target for state injection for rotations by $Z$, we can produce the $S=\exp(i \frac{\pi}{4} Z)$ gate.
This state injection protocol requires only $X,Z$, and CNOT gates and single qubit Pauli measurements, so it requires only the subgroup ${\cal C}$ above.  Note that
since we have the single qubit $Z=S^2$ available, we can recover state injection if we instead produce $S^\dagger$ (\cref{Sdistillation}).
\begin{align}\label{Sdistillation}
    MZ_2 \cdot CNOT  (|\psi\rangle \otimes |Y=+1\rangle) \to S^{\pm} |\psi\rangle
\end{align}
We can also use the same $Y=+1$ state in state injection to produce rotation by $\exp(i \frac{\pi}{4} X)$ (\cref{HSHdistillation}).
To do this we use the usual state injection protocol, except we interchange the control and target on the CNOT gate in state injection, and we replace the final $Z$ measurement in state injection with an $X$ measurement.  The effect of this is to perform state injection in a Hadamard basis (even though no Hadamards are used!), since $(H \otimes H) {\rm CNOT} (H \otimes H)$ is equal to a CNOT gate with control and target interchanged.
\begin{align}\label{HSHdistillation}
    MX_2 \cdot CNOT (|Y=+1\rangle  \otimes |\psi\rangle) 
\end{align}
$$\to HS^{\pm}H |\psi\rangle$$
Combining rotations $\exp(i \frac{\pi}{4} Z), \exp(i \frac{\pi}{4} X)$, we can produce the full single qubit Clifford group, including $H$.
\end{proof}

\subsubsection{$X,Z$ and $s/t$ is Universal}
\label{subsecxzst}
Let us now show finish the proof of \cref{lemxzst} by showing:
\begin{lemma}
Single qubit Pauli measurements can be derived using $X,Z$ and $s/t$.
\end{lemma}
\begin{proof} Recall that using just $X,Z$ and $s/t$, we have a protocol  $\hat O_{ZZ}$ which measures $ZZ$ on any two qubits $A,B$ and which measures nothing else with probability $1/2$.
So, imagine that we have a set $S$ of qubits such that the product $ZZ=+1$ for all pairs of qubits in $S$.  Consider another qubit $q$ in any state.  Then, pick any qubit $r$ from $S$ and measure $Z_q Z_r$.  With probability $1/2$ nothing else is measured and if $Z_q Z_r=+1$, we can add $q$ to the set; if $Z_q Z_r=-1$, we can apply $X$ to $q$ and then add $q$ to the set.  On the other hand, with probability $1/2$, additional information is measured; in this case, we remove $r$ from $S$.

So, with probability $1/2$, $|S|\rightarrow |S+1|$ and with probability $1/2$, $|S| \rightarrow |S-1|$.
This allows us to build up large sets $S$.  Since the size of $S$ does an unbiased random walk, it takes $\sim n^2$ operation to produce a set $S$ with $|S|=n$.  We can then use such a set as a ``standard": throw out any qubit in $S$.  The remaining qubits are then in a state which is an incoherent mixture of $|0\rangle^{\otimes (n-1)}$ and $|1\rangle^{\otimes (n-1)}$.  

So, let us analyze these two cases, where the remaining qubits are either in $|0\rangle^{\otimes (n-1)}$ or $|1\rangle^{\otimes (n-1)}$, separately.
First suppose that the state is 
$|0\rangle^{\otimes (n-1)}$.  Call each of these $n-1$ qubits ``$Z$ standards".
Then, any time we want to measure $Z$ on a single qubit, we get a $Z$ standard and use $\hat{O}_{ZZ}$ to measure $ZZ$ on the given qubit and the standard.  If we measure no other information other than $ZZ$, we can now in fact use both qubits as standards; if we also measure $XX$, we discard both qubits. 
So, at most one $Z$ standard is consumed per measurement. 
Note that even if $\hat{O}_{ZZ}$ does not succeed, so that we also measure $XX$, we still have learned the value of $Z$ on the given qubit: the effect is to measure the value of $Z$ on the given qubit and then bring that qubit and the standard into a Bell pair.

We assumed that the standards were in the state
$|0\rangle^{\otimes (n-1)}$.  If instead they are in the state
$|1\rangle^{\otimes (n-1)}$, nothing changes: our labels $|0\rangle$ and $|1\rangle$ are arbitrary and can be interchanged.  This is a consequence of a symmetry of our operations $X,Z,s/t$, which are invariant (up to an unobservable phase) under conjugation by $X$.  Indeed, if our labels were not arbitrary, then some sequence would reveal the difference between these two states, in which case we could, assuming we had 
$|1\rangle^{\otimes (n-1)}$, apply $X$ to every qubit to obtain $|0\rangle^{\otimes (n-1)}$.

So, we can measure single qubit $Z$ with a quadratic overhead: the number of operations is proportional to the square of the number of measurements.  Similarly, we can measure single qubit $X$ and single qubit $Y$ by preparing $X$ standards (respectively, $Y$ standards), which are sets of qubits where $XX=+1$ (respectively, $YY=+1$) for every pair in the set. This finishes the proof of the lemma and as a result \cref{lemxzst}.
\end{proof}
\begin{remark}
In fact, the quadratic overhead can easily be reduced to linear.  Suppose we have some such set $S$.  In $O(1)$ operations we can prepare another set $T$ with $|T|\geq 2$ such that all pairs of qubits in $T$ also have $ZZ=+1$.
Then, pick one qubit from $S$ and one from $T$ and apply $\hat O_{ZZ}$.  If this succeeds with $ZZ=+1$, add {\it all} qubits from $T$ to $S$ and if it succeeds with $ZZ=-1$, apply $X$ to all qubits in $T$ and then add all qubits in $T$ to $S$.
If it fails, remove the given qubit from $S$ and discard $T$.
Thus, with probability $1/2$ with have $|S|\rightarrow |S|+|T|\geq |S|+2$, while with probability $1/2$ we have $|S| \rightarrow |S|-1$.

This gives a biased random walk in $|S|$ and so $|S|$ increases linearly in the number of operations with high probability.  We leave it to the reader to consider optimizing the linear increase.  For example, what is the best $|T|$ to use; should one in fact apply this construction recursively by constructing $T$ using a similar process; one can in fact avoid discarding all of $T$ but only discard the measured qubit if $\hat O_{ZZ}$ does not succeed and so on. 
\end{remark}
\section{Permutational Quantum Computation}
\label{secpqc}
The model of permutational quantum computing \cite{jordan2010permutational} is as follows.

For any binary tree $T$ with $n$ leaves, we define a set of commuting operators on a system of $n$ qubits.
Each qubit corresponds to one leaf. For every vertex $v$ (including the root), there is an operator $\vec{S}_v^2 = (\sum_{l} \vec{S}_l)^2$ with eigenvalues corresponding to the total spin of the qubits corresponding to leaves which are descendants of that vertex. Additionally, for the root, there is another operator with eigenvalues corresponding to the total spin in the $Z$-direction of the qubits, denoted $S_Z = \frac{1}{2} \sum_{l \text{ leaf}} Z_l$.

The eigenvalues of these operators (all together forming a tuple of half-integers) define a complete eigenbasis. We say that a labelled tree is a binary tree with labels at each vertex (\cref{fig:labelledtree}), with the labels at each vertex being chosen from the set of eigenvalues of the operator(s) corresponding to that vertex.
We use $T,T',\ldots$ to denote unlabelled binary trees (\cref{fig:unlabelledtree}).  We use $\lambda,\lambda',\ldots$ to denote labeled binary trees, and use $|\lambda\rangle,\ldots$ to denote the corresponding states.
\begin{figure}[h]
    \centering
    \includegraphics[width=0.45\textwidth]{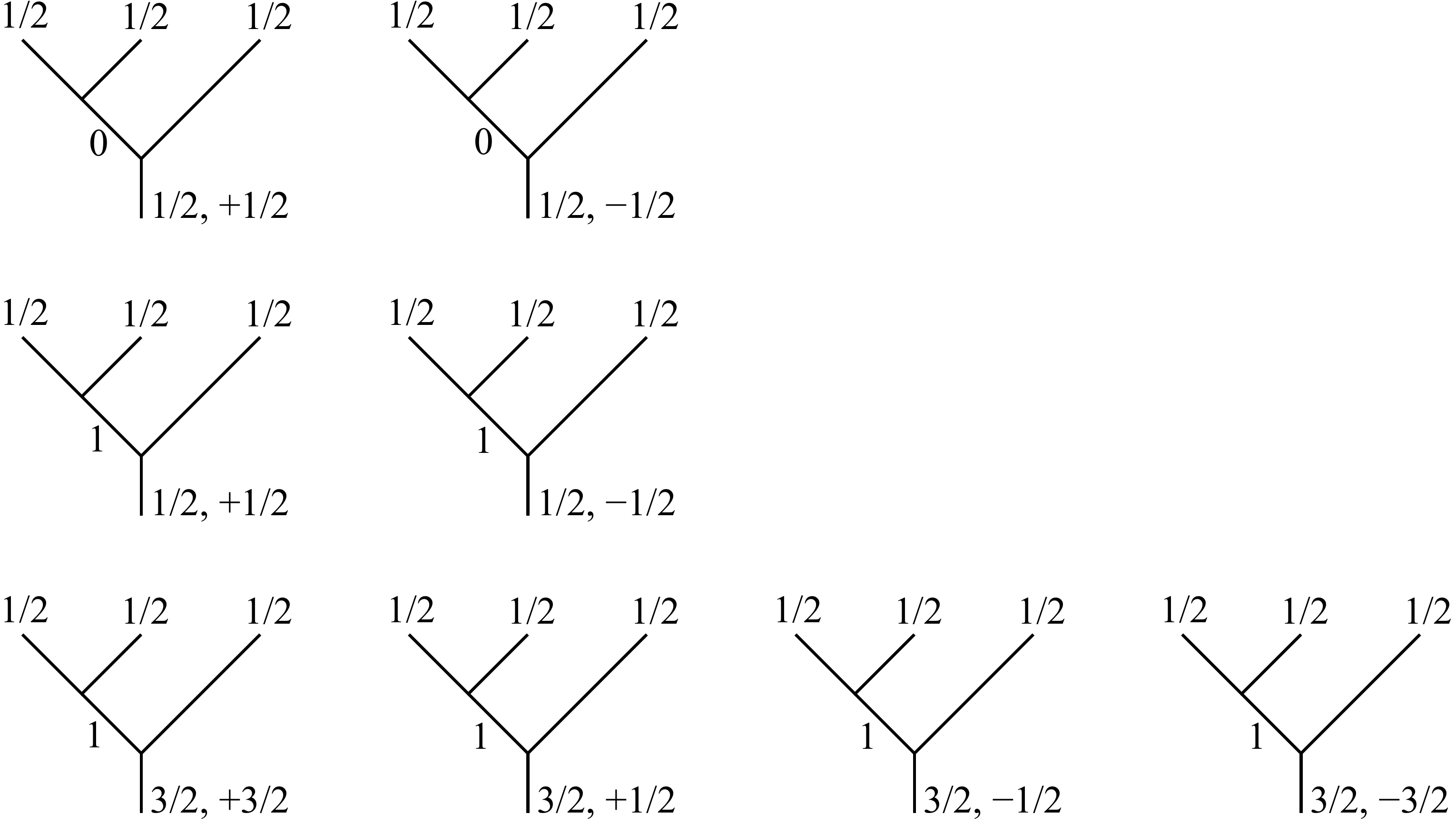}
    \caption{Eight labelled trees, corresponding to an orthonormal basis of $3$ spins, i.e. $(\mathbb{C}^2)^{\otimes 3}$ (Figure from \cite{jordan2010permutational}). Each label corresponds to the spins of the descendant leaves. However, the root has two coordinates; the first is the total spin of all leaves (total spin), and the second is the total azimuthal angular momentum, an eigenvalue of $S_Z$. Every labelled tree can be expressed in the standard basis of $(\mathbb{C}^2)^{\otimes 3}$. For example, the left tree from the second row is $|\lambda\rangle = \sqrt{\frac{2}{3}} |001\rangle - \frac{|010\rangle + |100\rangle}{\sqrt{6}}$.}
    \label{fig:labelledtree}
\end{figure}
\begin{figure}[h]
    \centering
    \includegraphics[width=0.45\textwidth]{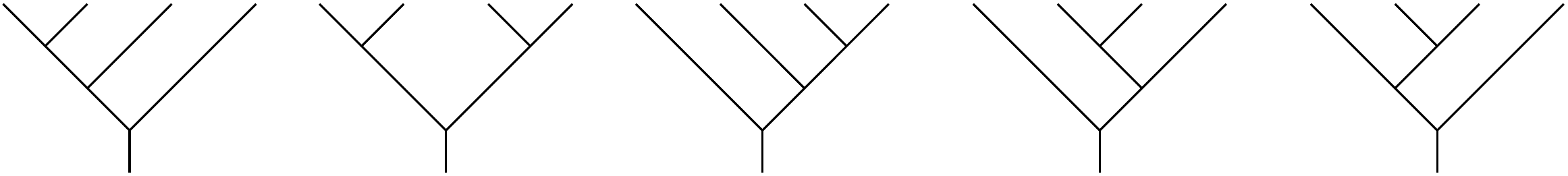}
    \caption{Examples of unlabelled trees ($T,T',\ldots$) for $4$ spins (Figure from \cite{jordan2010permutational}). Measuring, say, the leftmost unlabelled tree, is sampling a labelling of that tree called $|\lambda'\rangle$ with probability $|\langle \lambda|\lambda' \rangle|^2$ for $|\lambda\rangle$ a labelled tree with 4 spins.}
    \label{fig:unlabelledtree}
\end{figure}
There are two models of permutational quantum computing.
In the weak model, 
one can prepare any state $|\lambda\rangle$ corresponding to any labelled tree.
One can then pick any 
any other tree $T'$ and projectively measure the corresponding operators for that tree.
If $\lambda'$ is a labeling of $S$, then one can measure $|\langle \lambda | \lambda'\rangle|^2$ to inverse polynomial accuracy by repeating this projective measurement polynomially many times.
In the strong model, we assume that one can also measure $\langle \lambda | \lambda'\rangle$ (without taking the absolute value) to inverse polynomial accuracy.
While this problem of computing $\langle \lambda | \lambda'\rangle$ can be done classically \cite{H18}, there is no known efficient classical algorithm for the
problem of sampling from $\lambda'$ with probability $|\langle \lambda | \lambda'\rangle|^2$.

We show that:
\begin{theorem}
Using just $s/t$, one can simulate the weak model in polynomial time.
\end{theorem}

\subsection{Reducing to case that root has spin $0$}
First note that it suffices to consider the case that the root has spin $0$.  
The reader may prefer to skip this section on first reading.

Indeed, consider any $\lambda,\lambda'$ such that $\lambda$ has spin $S_{root} \neq 0$.
Introduce notation:  let $|\lambda,S_Z\rangle$ be a state defined by some labeled tree $\lambda$ with the $Z$-spin at root replaced by $S_Z$.
Then,
$\langle \lambda',S'_Z | \lambda,S_Z\rangle=0$ unless $S_{root}=S'_{root}$ and $S_Z=S'_Z$.
Fix some unlabelled tree $T'$.
To sample $\lambda'$ with probability $|\langle \lambda | \lambda' \rangle|^2$ , we do the following.
Adjoin an additional $2 S_{root}$ ancilla qubits, defining some tree $\lambda''$ which constrains those ancillas to be in the spin $S_{root}$ state (there are many possible such $\lambda''$ as all we care about is the root spin; choose an arbitrary such $\lambda''$).  Then join the root of $\lambda$ to the root of $\lambda''$ to produce\footnote{The root of $\lambda$ has two labels, since it also has a label by some $S_Z$.  This root becomes an internal vertex of $\hat \lambda$ which should have only one label.  We drop the label $S_Z$ when we build $\hat \lambda$.} some $\hat \lambda$ with spin $\hat S_{root}=0$.
Let $T''$ be an unlabelled tree, obtained by removing the labels from $\lambda''$.
Join the root of $T'$ to $T''$ to produce some $\hat T'$.
Then, prepare $|\hat \lambda\rangle$ and projectively measure the eigenvalues corresponding to tree $\hat T'$.
Note that we will always measure
$\hat S'_{root}=0$.  
This projective measurement gives a labeling of $\hat T'$, which induces a labeling of $T'$ since $T'$ is contained in $\hat T'$ (we can assume that the root of $T'$ has $S'_Z=S_Z$ since otherwise the amplitude is zero).

We claim that this correctly samples the amplitudes.  That is, if $\hat \lambda'$ is a labeling of $\hat T'$, which induces some labeling $\lambda'$ of $T'$, then
$|\langle \hat \lambda | \hat \lambda'\rangle|^2$  is equal to $|\langle \lambda | \lambda' \rangle|^2$.  To see that this, note that
$$|\hat \lambda\rangle=\sum_{S_Z=-S_{root}}^{+S_{root}} A(S_Z) | \lambda,S_Z\rangle \otimes |\lambda'',-S_Z\rangle,$$
where $A(S_Z)$ are some Clebsch-Gordan coefficients.  Do the same expansion for $|\hat \lambda' \rangle$ and then
$$\langle \hat \lambda' | \hat \lambda \rangle=\sum_{S_Z=-S_{root}}^{+S_{root}} |A(S_Z)|^2 \langle\lambda',S_Z|\lambda,S_Z\rangle.$$
Since the inner product on the right-hand side is independent\footnote{
We show this here.  Let (in a slight overuse of notation) $X$ denote the sum of Pauli $X$ operators on all qubits.
The state $\exp(i \theta X)|\lambda,S_Z=S_{root}\rangle$ is a superposition of states
$|\lambda,S_Z\rangle$ with different $S_Z$.
The inner product $\Bigl( \langle \lambda',S_{root}| \exp(-i \theta X)\Bigr) \Bigl(\exp(i \theta X) | \lambda,S_{root}\rangle \Bigr)$ is independent of $\theta$ and by expanding in $\theta$ it follows that
$\langle \lambda',S_Z| \lambda,S_Z\rangle$ is independent of the particular $S_Z$.
} of $S_Z$, the result follows.

\subsection{Measuring operators for some tree}
Next we describe how to projectively measure operators corresponding to any tree $T'$ with root having total spin $0$.  The key is that given any set $U$ of qubits with $|U|=M$ for some $M$, we can approximately measure the total spin of those qubits as follows.  Simply repeat the following operation many times: choose a random pair of qubits, and measure $s/t$.  
We show below that the number of measurements needed is at most inverse polynomial in the accuracy.

Then, letting $S^2$ denote the squared spin, and $\vec S_i$ denote a vector of spin operators (one-half the Pauli operators) on qubit $i$,
\begin{align}\label{totalspineq} 
S^2=\sum_{i=1}^M \sum_{j=1}^{M} \vec S_i \cdot \vec S_j
\end{align}
$$= \frac{3}{4} M + \sum_{i \neq j} \vec S_i \cdot \vec S_j $$
$$= \frac{3}{4} M + M(M-1)\frac{\sum_{i \neq j} \vec S_i \cdot \vec S_j}{M(M-1)}$$
Hence, averaging over randomly selected $i,j$, the expectation value of $\vec S_i \cdot \vec S_j$ equals
$(S^2-\frac{3}{4} M)/(M(M-1))$.  In a triplet state, $\vec S_i \cdot \vec S_j=1/4$ while in a singlet it equals $-3/4$.  So, the probability of triplet is $$(S^2-\frac{3}{4} M)/(M(M-1))+3/4.$$

Note that while this measurement of singlet or triplet may change the state of the qubits in $U$, it does not change the total spin and since we randomly choose $i,j$ each time, the probability of triplet depends only on the total spin.  Hence, since the measurements are independent, we may estimate the spin in a time which is inverse polynomial in the accuracy.
Indeed, since the spin is quantized, the convergence is actually faster than polynomial.  Once the number of measurements is polynomially large compared to $S^2$, the convergence in accuracy becomes exponential\footnote{This fact was used in \cite{rudolph2021relational} for the STP=BQP proof.}.

Using this ability to measure total spin of a set of particles as a primitive, we can measure the operators corresponding to any tree $T'$.  The key is to start at the leaves and work towards the root.  Measure operators corresponding first to vertices closest to the leaves.  Proceed toward the root, only measuring an operator on a given vertex once all operators below that vertex have been measured.

The point is, that while the operators corresponding to different vertices commute with each other, our measurement process reveals extra information.  Our measurement process need not commute for different vertices.  However, our measurement process for a given vertex does not affect the total spin on vertices closer to the root.

\subsection{Preparing states}
We now explain how to prepare a state $|\lambda\rangle$ in polynomial expected time.
In this case, we work in the reverse direction: we start at the root and work towards the leaves.  In particular, we will prepare a sequence of states $|\lambda_0\rangle,|\lambda_1\rangle,\ldots$.  
This corresponds to a sequence of {\it partially labeled} trees $\lambda_0,\lambda_1,\ldots$.  In a partially labeled tree, some set of vertices will have labels, and
if a vertex $v$ is labeled, then all vertices which are ancestors of $v$ will also be labeled.  The partially labeled tree $\lambda_0$ will have the root labeled with $S_{root}=0$ and no other labels.  Each tree $\lambda_{i+1}$ will have two more vertices labeled than the previous tree $\lambda_i$ in the sequence; this will be done by labeling a pair of vertices which are children of some vertex $v_i$ which is labeled in $\lambda_i$, and the last tree in the sequence will have all vertices labeled and will be the same as $\lambda$.

Each state $|\lambda_i\rangle$ corresponding to some partially labeled tree $\lambda_i$ will have the property that for every operator corresponding to a labeled vertex of $\lambda_i$, the state $|\lambda_i\rangle$ will have the corresponding expectation value for that operator.  No other properties of the state are assumed, so a partially labeled tree does not uniquely specify a state.

Preparing the first state $|\lambda_0\rangle$ is easy: one may simply create any choice of $n/2$ singlets.
We will 
construct a primitive operation that we call {\it splitting} which has the following properties. {\bf 1}: The splitting primitive takes as input a 
set of $n$ qubits with total spin $S$ for some $n,S$; there are no other assumptions on the input state.
{\bf 2}: One fixes some $m<n$ and some $S', S''$
with $|S''-S'| \leq S \leq S''+S'$.
{\bf 3}: The splitting primitive applies $s/t$ measurements to qubits from that set, taking polynomial expected time.
{\bf 4}: The resulting state has (up to exponentially small error) the first $m$ qubits with total spin $S'$
and the remaining $n-m$ qubits with total spin $S''$.

Given this splitting primitive, we can then produce each state $|\lambda_{i+1}\rangle$ from the state $|\lambda_{i}\rangle$, by applying splitting to the set of qubits corresponding to descendants of $v_i$, with $S',S''$ depending on the labels in $\lambda_{i+1}$ for the children of $v_i$.

We now construct the splitting primitive above.
First, let us say that a set of $n$ qubits with total spin $S$ are in {\it canonical form} if
there are $n-2S$ singlet pairs (in some fixed configuration, rather than a superposition) and the remaining $2S$ qubits are in a totally symmetric state.  For example, a state on $8$ qubits with qubits $1,3$ in an singlet and $4,7$ in a singlet and $2,5,6,8$ in a totally symmetric state is in canonical form.

We divide the construction of splitting into four steps.

{\it First Step---}
First, we take the $n$ qubits and bring them to canonical form.
The construction is recursive.
If $n=2S$, then they are already in canonical form.  If not, pick a pair of qubits at random and measure $s/t$.  If the result is triplet, pick another pair at random, and try again, continuing until eventually some pair is in a singlet.   That gives one singlet; we then bring the remaining $n-2$ qubits to canonical form using the same algorithm recursively.  This takes polynomial expected time.

{\it Second Step---}
Recall that we wish to divide the set of $n$ qubits into two sets, with $m$ and $n-m$ qubits respectively, and with total spin $S'$ and $S''$, with $S''\geq S'$.
Let $\Delta=S''-S'$.
Let $S'_{min}$ and $S''_{min}$ be the two values for total spin such that $S''_{min}-S'_{min}=\Delta$ and
$S''_{min}+S'_{min}=S$.

Our second step will be to take the state after the first step which is already in canonical form, and divide it into two sets of qubits, of sizes $m,n-m$ respectively, with total spins $S'_{min}$ and $S''_{min}$ respectively, with each set in canonical form.
Call these two sets $Q_1$ and $Q_2$.
This division can be done easily: take the $n-2S$ qubits in a totally symmetric state, and place $2S'_{min}$ qubits in $Q_1$ and the remaining $2S''_{min}$ in $Q_2$.  Then, take the singlets from the state in the first step, and place each singlet into one of the two sets (either $Q_1$ or $Q_2$), so that the total sizes of the two sets are correct.

{\it Third Step---}
Our third step acts only on certain subsets of sets $Q_1,Q_2$.  We call these subsets $R_1,R_2$ and both have size $2\Delta$.  These subsets $R_1,R_2$ are given by choosing $\Delta$ singlets from $Q_1$ and letting those
 be in $R_1$ and choosing $\Delta$ singlets from $Q_2$ and letting those  be in $R_2$.
 Qubits will remain in the subset they are in after the second step, but their state will change due to this step.  What the step will do is bring it to a state where those qubits are now in a totally symmetric state in each set individually (i.e., the $2\Delta$ qubits in $R_1$ are totally symmetric, as are the $2\Delta$ qubits in $R_2$), while the total spin of those $4\Delta$ qubits is still $0$.

To do this, we use an iterative algorithm.  First, {\bf A}: bring both subsets $R_1,R_2$ into canonical form.  If there are no singlets in either one, terminate.  Otherwise, {\bf B}: if either subset has a singlet, then so must the other\footnote{This is because the total spin of qubits in $R_1,R_2$ is $0$.  If $R_1$ has at least one singlet, then $R_1$ has total spin less than the maximal value of $\Delta$ and so $R_2$ must also have total spin less than $\Delta$ in order for total spin of $R_1,R_2$ to be $0$.}.  Pick a pair of qubits, one from a singlet in each subset, and measure $s/t$.  Call these qubits $1,2$ and call the singlets that they are in respectively $s_1,s_2$.  Then, regardless of the measurement outcome, measure $s/t$ on both qubits in $s_1$.  If the result is $s$, this pair of measurements has had no effect on the spins; in this case repeat the pair of measurements, continuing until the second measurement give $t$.  Once the second measurement gives $t$,
go back to step {\bf A}.
 
We claim that this takes polynomial expected time.  Consider the effect of {\bf B}.  Suppose that before measuring {\bf B} the total spin in each set was $S$.  After {\bf B},
the total spin in each set must be one of the possibilities $S-1,S,S+1$.  The total spin squared\footnote{We remind the reader that the spin squared is $S(S+1)$, not $S^2$.} before {\bf B} is $S(S+1)$ and the total spin squared after {\bf B} has expectation value $S(S+1)+2$.  
The process of bringing into canonical form measures the total spin in each set.
We claim that the probability that the total spin squared is $S+1$ is greater than the probability that
it is $S-1$.  Indeed, this must follow in order for the total spin squared to have
expectation value $S(S+1)+2$ after {\bf B} because
$$\frac{(S-1) S + (S+1) (S+2)}{2}=S^2+S+1$$
$$<S(S+1)+2.$$
Hence, the total spin does a biased random walk with the bias toward increasing spin, and so the spin must become maximal in at most polynomial time.

{\it Fourth Step---}
After the third step, we have the following.
Each of the two sets $Q_1,Q_2$ has three subsets.  Let $1A,1B,1C$ denote the three subsets of the first set and $2A,2B,2C$ denote the three subsets of the second set, with $1A,1B$ comprising $R_1$ and
$2A,2B$ comprising $R_2$.   The sets $1A,2A$ each contain qubits in some product of singlets.
The sets $1B$ and $2B$ each contain qubits in a totally symmetric state, with the union of $1B$ and $2B$ having total spin $0$.  The sets $1C$ and $2C$ also each contain qubits in a totally symmetric state, but now the union of $1C$ and $2C$ is also in a totally symmetric state.
In the fourth step, we act on sets $1B$ and $1C$ to try to bring them to a totally symmetric state (we also do the same procedure to $2B$ and $2C$ with the same goal but we just describe it for $1B,1C$).  To do this, we we apply a large number of $s/t$ measurements on qubits randomly chosen from the union of $1B$ and $1C$.  If all measurement outcomes are $t$, then the application of these measurements converges to projecting onto the state where those qubits in $1B$ union $1C$ are in a totally symmetric state of total spin $S'$  and we succeed.  The convergence to this projector is exponential, once one has more than polynomially many measurements.  

The probability that all measurements are $t$ in this step is at least inverse polynomial.  This may be seen by computing the projection of the initial state onto the space where $1B,1C$ are totally symmetric and $2B,2C$ are totally symmetric.
If we fail, so that some measurement is $s$, we repeat all steps of this process.

\section{Generalizing the Model}
\label{secgtm}
The model STP can be generalized in several ways.  One natural generalization is to consider symmetries other than $SU(2)$, such as $SU(m)$ for $m>2$.  Another natural generalization is to consider higher spin representations of $SU(2)$.

In such a higher spin representation model, we have qudits, for $d=2S+1$ with $S$ integer or half-integer.  We may consider having several different kinds of qudits simultaneously, for example having both qubits ($S=1/2$) and qutrits ($S=1$).  We allow any two qudits (perhaps of different dimensions) to be brought together, and the total spin to be measured.

As a toy physical example, one might imagine deuterium.  The deuterium nucleus has total spin $1$, while the electron has spin $1/2$.  The deuterium atom then has total spin $1/2$ or $3/2$ and there is a hyperfine splitting between these states.

Interestingly, this higher spin model can be simulated using just $s/t$ on qubits.  To simulate a qudit with spin $S$, use $2S$ qubits in a totally symmetric state.  When two qudits with spin $S,S'$ are brought together, we can measure the total spin of the $2S+2S'$ qubits by repeatedly selecting pairs of qubits uniformly at random and measuring $s/t$.  Exactly as in \cref{secpqc}, the convergence in accuracy is exponential once we have a number of measurements which is polynomial in the total spin.

After measuring the total spin, we can then bring the $2S+2S'$ qubits into two sets of $2S,2S'$ qubits respectively, both in a totally symmetric state, again as in \cref{secpqc}.  Indeed, this is done by the splitting primitive.

\section{PostSTP Equals PostBQP}
\label{postbqp}
Let us formally define STP with postelection.  
\begin{definition}
Let PostSTP be the class of languages $L\subset \{0,1\}^{*}$ such that for all inputs $x$ the following holds.  A quantum state is initialized as a product of
polynomially many singlets.
Then some classical algorithm (determined by $L$) takes $x$ as input and outputs a sequence of polynomially many $s/t$ measurements (as well as outcomes to postselect upon for all but the last measurement), taking polynomial time to output this sequence.  The sequence of measurements is applied to the input state.
Postselection is applied on all but the last measurement,  with the promise that the outcomes postselected on have nonzero probability.  Finally, if $x\in L$, then the last measurement is $s$ with probability at least $p$ for some $p>0$ and if $x\not \in L$, then the last measurement is $s$ with probability at most $p'$ for some $p'$ strictly less than $p$.  The quantities $p,p'$ are independent of input size.
\end{definition}
\begin{remark}
To define postselection in symbols, let $\psi_0$ be the product of singlets which is the initial state.
Let $\Pi_j$ be a projector which projects onto the desired outcome (either $s$ or $t$) of the $j$-th measurement.
Let $A_j=\Pi_j \Pi_{j-1} \ldots \Pi_1$ be the product of these projectors up to the $j$-th one.
Let there be $K=\poly(N)$ measurement outcomes that we postselect on, and let $\Pi_{k+1}$ denote the projector onto the singlet outcome for the last (the $(K+1)$-st) measurement.
Then, the probability that the last outcome is $s$ is
$$\frac{\langle \psi_0 | A_K^\dagger \Pi_{K+1} A_K | \psi_0\rangle}{\langle \psi_0 | A_K^\dagger A_K | \psi_0\rangle},$$
where we are promised that the denominator is nonzero.
\end{remark}

We will show
\begin{theorem}
PostSTP equals PostBQP.
\end{theorem}

In outline, we prove this result by first showing  in \cref{ite} that we can use postselection to simulate imaginary time evolution with Heisenberg interactions.  Ref.~\cite{C16} showed that approximating the ground state energy of Hamiltonian with Heisenberg interactions is QMA-hard (see also \cite{Cubitt_2018}), so this shows already PostSTP is at least as powerful as QMA: simply evolve under the desired Heisenberg Hamiltonian for a polynomially long imaginary time.
To show that we get PostBQP, we use the ability to produce evolution ``time-dependent" Heisenberg interactions in imaginary time, i.e., to vary the Heisenberg Hamiltonian that we evolve under.  We use the same encoding as in \cite{C16} to show in \cref{postbqpsubsec} that this gives us PostBQP.  Appropriate choices of time-dependent Heisenberg Hamiltonians will give us both circuits and measurement.  We will have to pay some attention to making errors exponentially small when we do this.

Before doing this, we show in \cref{smallamp} that the probabilities that we postselect on can only ever become exponentially small in $\poly(N)$.

\subsection{A Remark On How Small Amplitudes Can Become}
\label{smallamp}
Let $\psi_0$ be the product of singlets which is the initial quantum state for PostSTP.

We show
\begin{lemma}
The nonzero probabilities that can occur in PostSTP are all $\Omega(\exp(-\poly(N))$.
\begin{proof}
The probability that the $j$-th measurement outcome has the desired value is
$$\frac{\langle \psi_0 | A_{j}^\dagger  A_{j} | \psi_0\rangle}{\langle \psi_0 | A_{j-1}^\dagger A_{j-1} | \psi_0\rangle}.$$
Clearly, the denominator is at most $1$.  So, we lower bound the numerator.

Note that every $t$ postselection can be replaced by $(1/2)(1+{\rm SWAP})$, where ${\rm SWAP}$ is the gate that swaps a pair of qubits.  Similarly, every $s$ postselection can be replaced by $1/2(1-{\rm SWAP})$.
Hence, the numerator is the sum of $4^{j}$ different terms, corresponding to replacing individual projectors $\Pi_j$ by either the identity or ${\rm SWAP}$.
The contribution of any such term to the expectation value is of the form
$4^{-j} \langle \psi_0 | {\rm Permute} | \psi_0\rangle$, where ${\rm Permute}$ is the composition of some SWAPs, hence some permutation to the qubits.
Recall that $\psi_0$ is the product of singlets, and therefore it is easy to see that for any permutation Permute of the qubits, the expectation value $\langle \psi_0 | {\rm Permute} | \psi_0\rangle$ is $\pm 2^{-J}$ for some $0\leq J \leq N-1$.
So, the expectation value is a sum of $4^{j}$ different terms, each of which equals $\pm 4^{-j} 2^{-J_{\rm Permute}}$, for some $J_{\rm Permute}$ depending on the permutation Permute.  Hence, if nonzero, the expectation value is at least
$4^{-j} 2^{-N-1}$, where $j\leq \poly(N)$.
\end{proof}
\end{lemma}

\subsection{Simulating Imaginary Time Evolution}
\label{ite}
Let $\psi_{resource}(\epsilon)$ be the four qubit state, on four qubits called $C,D,E,F$, with total spin $0$ given by
first preparing $C,E$ in a singlet and $D,F$ in a singlet and then acting with $1+\epsilon \vec S_E \cdot \vec S_F$, and finally appropriately normalizing the state to unit norm.

Now consider the effect of bringing in two extra qubits $A,B$ in some arbitrary state $\psi_{AB}$, and postselecting on $A,C$ being $s$ and on $B,D$ being $s$.
One may commute the postselection on $A,C$ and $B,D$ being in $s$ with the action of 
$1+\epsilon \vec S_E \cdot \vec S_F$ used to define $\psi_{resource}(\epsilon)$.
So,
the effect of this operation is to teleport the state $(1+\epsilon \vec S_A \cdot \vec S_B) \psi_{AB}$ to qubits $E,F$, while leaving $A,C$ and $B,D$ in singlets.

So, if we can produce $\psi_{resource}(\epsilon)$, we can apply the operation $1+\epsilon \vec S_A \cdot \vec S_B$.
Of course, the ability to apply $1+\epsilon \vec S_A \cdot \vec S_B$ correspondingly implies the ability to
create $\psi_{resource}(\epsilon)$.  Thus, the two things (the state and the operation) are equivalent as resources.

Note that the projection onto $t$ is given (up to normalization) by the operation $(1+\epsilon_0 \vec S_A \cdot \vec S_B)$ with
$\epsilon_0=4/3$, so we can create $\psi_{resource}(4/3)$.

We will give a protocol to reduce $\epsilon$, consuming a pair of states $\psi_{resource}(\epsilon)$ with given $\epsilon$ and applying the operation $(1+\epsilon' \vec S_A \cdot \vec S_B)$ to an arbitrary pair of qubits $A,B$,
with
\be
\label{epsp}
\epsilon'=\frac{\epsilon^2}{4}.
\ee
If $A,B$ were in a singlet initially, this creates $\psi_{resource}(\epsilon')$.
By applying this protocol repeatedly, we can produce a sequence of states $\psi_{resource}(\epsilon_i)$ with $\epsilon_0=4/3$ and
$$\epsilon_{i+1}=\frac{\epsilon_i^2}{4}.$$
\begin{remark}
Of course we could also start with $\epsilon_0=-4$ which projects onto $s$ up to normalization, but this does not lead to anything interesting.
\end{remark}

The protocol is as follows.  Create a pair of qubits $C,D$ in a singlet, and let $A,B$ be arbitrary.  Then, apply $(1+\epsilon \vec S_A \cdot \vec S_C)(1+\epsilon \vec S_B \cdot \vec S_D)$, consuming the two copies of $\psi_{resource}(\epsilon)$ to do this, and again project onto $C,D$ in a singlet.  A little algebra shows that the resulting state (up to normalization) is a singlet on $C,D$ and the operation
$(1+\epsilon' \vec S_A \cdot \vec S_B)$ is applied.
We emphasize that \cref{epsp} is {\it not} a perturbative result for small $\epsilon$, but rather holds for {\it all} $\epsilon$.

The cost of applying the $i$-th operation, $(1+\epsilon_i \vec S_A \cdot \vec S_B)$, is exponential in $i$.  At the same time, the magnitude of $\epsilon_i$ decreases doubly-exponentially in $i$, roughly squaring at every step.
So, for any $x\in (0,1]$, we can construct the operation $1+\epsilon \vec S_A \vec S_B$ for some $\epsilon$ in the interval $[x^2,(4/3)x]$
using a number of operations at most logarithmic in $\epsilon^{-1}$: simply search for the first term in the sequence which lies in this interval.

To obtain operations with negative $\epsilon$,
consider a slight modification of the above sequence of operations. Create a pair of qubits $C,D$ in a singlet, and let $A,B$ be arbitrary.  Apply $(1+\epsilon_0 \vec S_A \cdot \vec S_C)(1+\epsilon_i \vec S_B \cdot \vec S_D)$, and again project onto $C,D$ in a singlet, with $\epsilon_0=-4/3$ and $\epsilon_i$ from the above sequence. So,
\begin{lemma}
\label{sij}
For any $\tilde \epsilon\in [-1,1]$ and any $\delta>0$, we can construct the operation
$1+\epsilon \vec S_A \cdot \vec S_B$ for some $\epsilon$ with $|\epsilon-\tilde \epsilon|\leq \delta$, using a number of operations at most $O(\log(\delta^{-1}) \delta^{-2})$.
\end{lemma}
\begin{proof}
Construct $(1+x \vec S_A \cdot \vec S_B)$ for some $x$ which has the same sign as $\tilde \epsilon$, with
$x$ sufficiently small compared to $\delta$, but $x$ not much smaller than $\delta^2$, i.e. $x=\Omega(\delta^2)$.  Then, take powers of this operation to obtain a suitable $1+\epsilon \vec S_A \cdot \vec S_B$.  Since $x$ is not much smaller than $\delta^2$, it takes at most $O(\delta^{-2})$ operations to do this.
\end{proof}
\begin{remark}
When constructing powers of the operation, it is more convenient to work with an exponentiated form of the operation: $1+x \vec S_A \cdot \vec S_B=\exp(y \vec S_A \cdot \vec S_B)$, where $y$ depends on $x$.
\end{remark}
\begin{remark}
Likely the dependence of the number of operations on $\delta$ can be greatly reduced.  We have given an inefficient construction that involves first constructing an operation with a very small $x$ and then taking powers of that operation.  One might instead apply operations with several different $\epsilon_i$ with different magnitudes to reduce the total.  We do not worry about that here.
\end{remark}
So,
\begin{lemma}
\label{lemmaite}
In PostSTP, we can approximate imaginary time evolution under a time-dependent Heisenberg Hamiltonian, up to inverse polynomial error.
\end{lemma}
The proof of this lemma is immediate, once we define what we mean.  A ``Heisenberg Hamiltonian" means a Hamiltonian of the form $H=\sum_{i,j} J_{ij} \vec S_i \cdot \vec S_j,$ where $J_{ij}$ is some polynomially-bounded matrix.
By imaginary time evolution under such a Hamiltonian, we mean evolving an initial state under the equation
$\partial_t \psi=H \psi$, up to normalization.  By ``time-dependent", we allow $H(t)=\sum_{i,j} J_{ij}(t) \vec S_i \cdot \vec S_j,$ where $J_{ij}(t)$ is some polynomially bounded matrix which depends on $t$ and we allow
$\partial_t \psi(t) = H(t) \psi(t)$.  Finally, the inverse polynomial error is an error on the right-hand side of the evolution equation:
$\partial_t \psi=H(t) \psi + E(t),$ where $E(t)$ is a state whose norm may be made polynomially smaller than $\psi(t)$.

Then, to prove the lemma, we simply Trotterize\footnote{The evolution $e^{t (\sum_i H_i)}$ can be efficiently approximated by the Lie-Trotter product formula $(\prod_i e^{tH_i/n})^n$ for $n$ small enough.} the imaginary time evolution equation, and simulate the Trotter steps using postselection by picking $\epsilon$ to be polynomially small in the state $\psi_{resource}(\epsilon)$, and applying \cref{sij}.

This already immediately implies that
\begin{corollary}
PostSTP contains QMA.
\begin{proof}
By \cite{C16}, approximating the ground state of a Heisenberg Hamiltonian to inverse polynomial error is QMA-hard. The operator ${\rm exp}(-Ht)$ is a ground state projector and hence by approximating this imaginary evolution up to inverse polynomial error, we can estimate the ground state of this Hamiltonian.
\end{proof}
\end{corollary}

\subsection{Simulating PostBQP}
\label{postbqpsubsec}
However, our goal is not just to prove that PostSTP contains QMA, but rather than PostSTP is equivalent to PostBQP.

To do this, we again use results from \cite{C16}.  We use the result (see section 5.1 of that paper) that using Heisenberg interactions we can implement logical qubits (using three physical qubits for each logical) and obtain terms $X,Z$ on any given logical qubit as well as terms $XX,ZZ$ on pairs (see also \cite{DiVincenzo_2000} for an encoding using the Heisenberg interaction).
So, it suffices then to consider a model of computation in which we have qubits (which for the rest of this subsection refer to the {\it logical} qubits of \cite{C16}), and the ability to implement imaginary time evolution under time-dependent $X,Z,XX,ZZ$, up to inverse polynomial error.  Throughout this subsection, when we refer to time, we mean imaginary time.

By turning on an $XX$ term on a pair of qubits for a long time and then turning it off (while leaving other terms off), we can approximately project  onto the $XX=+1$ or $XX=-1$ eigenstates, and similarly using $ZZ$ terms we can approximately project onto $ZZ=+1$ or $ZZ=-1$ eigenstates.  Further, by turning on a sum of $X$ and $Z$ terms on a single qubit (while leaving other terms off), we can prepare an ancilla qubit in a state which is approximately any desired pure state $\cos(\theta) |0\rangle + \sin(\theta)|1\rangle$.

These abilities suffice.  First, we can approximately prepare a qubit in a $|+\rangle$ state, and then using it as an ancilla, we can approximately apply a CNOT from a source to a target by using the ability to approximately postselect on $ZZ=+1$ and $XX=+1$ eigenstates.  This is the same as we used in \cref{ifcg} from \cite{Xue_2013,Karzig_2017}.
Note importantly that here we are using imaginary time evolution to approximately apply a unitary gate.  This should be not be surprising; after all, the well-known idea of measurement based quantum computation \cite{Raussendorf_2003} uses a sequence of measurements to apply unitary gates.

Also, by preparing ancillas which are approximately in states $\cos(\theta) |0\rangle + \sin(\theta)|1\rangle$ for other choices of $\theta$, and using the CNOT gates and state injection, we can approximately implement unitary rotations to produce approximate rotations by $\exp(i \theta Y)$ where $Y$ is the single qubit Pauli $Y$.

The ability to do CNOT and rotations $\exp(i \theta Y)$ allows universal quantum computation.  
So, if we ignore issues of error (i.e., the fact that all our constructions only approximately gave these gates), we can implement arbitrary quantum circuits and further we can postselect on measurement outcomes since we can implement imaginary time evolution under $Z$, giving PostBQP.

The issue of error can be resolved by implementing a fault tolerant construction using these approximate gates. 
Note that the error in the gates can be made arbitrarily small.  Indeed, we can even make the error in individual gates $1/\poly(N)$ for any polynomial at the cost of a polynomial overhead, so we can make the error in gates polynomially smaller than the the inverse of the total number of gates!  Thus, if our goal were simply to simulate BQP, where we would be satisfied with an inverse polynomial error in the output probabilities of our quantum circuit, there would be no need for any fault tolerance.  However, since we want to simulate PostBQP, we need exponentially small errors.  So, some fault tolerant construction is needed.

Of course, the ability to make the error in individual gates polynomially small compared to the inverse of the total number of gates certainly simplifies the fault tolerant construction.  However, we claim that in fact the usual threshold theorems can be applied to our setting and so long as the error in individual gates is sufficiently small, we can make the error in logical operations exponentially small.

The usual threshold theorems involve replacing idealized unitary gates by CPTP maps that approximate (in diamond norm) the desired unitary operations.  Instead, we are replacing idealized unitary gates by linear operators that act on pure states (rather than mixed states) that are close in operator norm to the desired unitary.  However, we now show that the usual threshold theorems can be adapted to this case.

Consider any given gate in the circuit, which ideally would be implemented by some unitary we will call $U$.  Suppose instead we implement some map $A$ (which is a linear operator on pure state) with $\Vert A-U \Vert \leq \epsilon$, where $\Vert \cdots \Vert$ is the operator norm.  Then, $(1-\epsilon)A$ has singular values bounded by $1$ and if we define $B$ to be any operator such that $B^\dagger B =I-(1-\epsilon)^2 A^\dagger A$ then
$${\cal E}(\rho) \equiv (1-\epsilon)^2 A \rho A^\dagger + B \rho B^\dagger$$
is a CPTP map with two Krauss operators, $(1-\epsilon) A,B$.

Making such a replacement for all gates in the circuit, we get a ``noisy circuit" where every gate has two Krauss operators, with the first Krauss operator being $O(\epsilon)$ close to the ideal unitary.  The usual threshold theorems apply to this noisy circuit for sufficiently small $\epsilon$ showing that the logical error is exponentially small.  Then, the situation relevant for us is one in which every time a CPTP map is applied for some gate, we pick the first Krauss operator (that is, $(1-\epsilon) A$ in this case), rather than the second.  
This can be physically thought of as selecting a particular noise pattern.

So, we ask: if we select the first Krauss operator for every CPTP map in the noisy circuit, is the error still exponentially small?  Intuitively, there is no problem: the first Krauss operator is closer to what we want than the second Krauss operator, so the situation in which we select that Krauss operator every time should be even better than a random choice.

To prove this, of course one could reprove the threshold theorem, using our gates (which are not noisy in that they map pure states to pure states but which still only approximate the desired gates).  Indeed, this could be done simply by following through some existing construction and showing that the error is exponentially reduced if the first Krauss operator is selected every time.  However, we would like to minimize our effort, and show that the error is exponentially small using the standard threshold theorem as a ``black box", without delving into the proof.
To do this, we use a simple trick.  Use some hierarchical construction to prove the threshold theorem, see for example \cite{Gottesman_2010} for a review of various constructions.  Using such a hierarchical construction of the threshold theorem, one considers codes of some $O(1)$ size, and proves that each step of the hierarchy reduces the error rate.  If the error rate is $\epsilon$ at some level, then it is $O(\epsilon^k)$ for some $k>1$ at a higher level of the hierarchy and so for sufficiently small $\epsilon$ the error rate reduces. 
We regard this error rate as a sum over different trajectories, corresponding to different choices of Krauss operators in each CPTP map.  The contribution of the trajectory where we select the first Krauss operator in every map is $1-O(\epsilon)$, with the constant hidden in the big-O notation depending on the size of the code.  So, for sufficiently small $\epsilon$, the contribution of the that trajectory is the dominant one and so the error rate for that trajectory at the next level of the hierarchy is also $O(\epsilon^k)$, i.e., the hierarchical construction reduces error rate for sufficiently small $\epsilon$ for this case as well.

\begin{remark}
Ref.~\cite{Cubitt_2018}  shows how to simulate 2D topological phases from imaginary time Heisenberg evolution. Anyons may be built into these simulations as boundary conditions, and time dependence of the Hamiltonians allows these anyons to be braided adiabatically, and then, later, collective (topological) charge states to be measured by simulating interferometry. In other words, once one can build imaginary time Heisenberg evolution (as in \cref{lemmaite}), one can simulate the complete operation of a topological quantum computer. This provides an alternative path from \cref{lemmaite} to showing that PostSTP equals PostBQP,  the post-selection in the conclusion stemming from our ability to post-select measurement outcomes of the simulated topological quantum computer, and the topological protection of the state allowing us to make errors exponentially small.  While this is an extremely concise argument in outline, some details regarding efficiency should be filled in; this is why we gave the more explicit argument first.
\end{remark}

\begin{acknowledgements}
We would like to thank Dave Bacon for pointing out the relevance of \cite{Rudolph_2005}, and anonymous referees for helpful comments to improve our paper. The first named author would like to thank the Aspen Center for Physics for their hospitality. The third named author would like to acknowledge the support of the Perimeter Institute for Theoretical Physics and Microsoft. Research at Perimeter Institute is supported by the Government of Canada through Innovation, Science and Economic Development Canada and by the Province of Ontario through the Ministry of Research, Innovation and Science. The experiments were conducted using Microsoft computational resources. 
\end{acknowledgements}

\bibliographystyle{plainnat}
\bibliography{st}
\onecolumn\newpage
\appendix

\section{Approximating rotations by powers of a fixed rotation}
\label{irrmeasure}
Let $\theta$ be an irrational multiple of $2\pi$.  Then, the integer multiples of $\theta$ are dense in the unit circle.  This has the implication that, given the ability to implement rotation $\exp(i Z \theta)$ on a single qubit, one can approximate the rotation $\exp(i Z \phi)$ for any $\phi$ to any desired accuracy by repeating the operation $\exp(i Z \theta)$ sufficiently many times.

However, it becomes interesting to ask: {\it how many} times must one repeat the operation $\exp(i Z \theta)$ to attain some given accuracy. This is equivalent to the question: given some $\phi$ and some $\delta>0$, what is the smallest $n$ such that $n \theta$ is within distance $\delta$ of $\phi$, modulo $2\pi$.

The answer to the question depends on the {\it irrationality measure} $\mu$ of $\theta/(2\pi)$.  The irrationality measure $\mu$ of a number $x$ is defined to be the smallest number such that for any $\epsilon>0$, we have
$$|x-\frac{p}{q}|>\frac{1}{q^{\mu+\epsilon}}$$
for all sufficiently large integers $q$, for all $p$.

Some numbers have infinite irrationality measure; these are called Liouville numbers.  As an example, consider a number such as $\sum_{j>0} 10^{-f(j)}$, where $f(j)$ is some fast growing function such as $j!$.  Almost all (in the sense of measure theory) real numbers have irrationality measure $2$.  The arctangent of any algebraic number has finite irrationality measure \cite{imeas}.

\begin{lemma}
Assume $\theta/(2\pi)$ is irrational with irrationality measure $\mu$.
Let us define the distance between angles in the obvious way, as the shortest distance between them on the circle; we write this distance using an absolute value sign.

Then, for any $\epsilon>0$, for any $\phi$ and any $\delta>0$, there is some $m$ with
$|m\theta-\phi|\leq \delta$ and with
$m\leq O(\frac{1}{\delta^{\mu+\epsilon}})$, where the constant hidden in the big-O notation depends on $\epsilon$.
\begin{proof}
Consider the set $\{0,\theta,2\theta,\ldots,k\theta\}$, where $k=\lceil 2\pi /\delta \rceil$.
All elements of this set are distinct.  So, there must be two different elements, call them $j_1 \theta,j_2\theta$, with $|j_1 \theta-j_2\theta|\leq (2\pi)/k$.
Let $q=j_2-j_1$, so $q\theta$ is within distance $(2\pi)/k\leq \delta$ of $0$.

By the definition of irrationality measure, multiplying both sides by $2\pi q$, we get
$|q\theta-2\pi p|> (\frac{2\pi}{q})^{\mu+\epsilon-1}$ for sufficiently large $q$.
So, for sufficiently large $q$, we have $q\theta$ at least distance $(\frac{2\pi}{q})^{\mu+\epsilon-1}$
from $0$.  For smaller $q$, it is possible that $q\theta$ is closer than that, but there are only finitely many such $q$, so $q\theta$ is at least $\theta_{min}$ for some $\theta_{min}>0$ depending on $\epsilon$.

Now consider the sequence of angles $0,q\theta,2q\theta,3q\theta,\ldots$ .  These angles change by at most $\delta$ from one to the next, so some angle in the sequence is within $\delta$ of $\phi$.

At the same time, they change by at least $\Delta\equiv{\rm min}(\theta_{min},(\frac{2\pi}{q})^{\mu+\epsilon-1})$, so we only
need to consider terms up to $\lceil2\pi \Delta^{-1}\rceil$ in the sequence to get an angle within $\delta$ of $\phi$.  Note that $q$ is bounded by $2k$, so it suffices to consider terms up to $O(k^{\mu+\epsilon-1})$ in the sequence.  Since $q$ is $O(k)$, this means we consider $m$ up to $O(k^{\mu+\epsilon})$.
\end{proof}
\end{lemma}

\section{No-leakage and equiangular sequences}

Our paper has antecedents in a previous work \cite{fisher2015quantum}, where the possibility of \textit{chemically} protected quantum computation \cite[section 8]{freedman2019quantum} was introduced; the proposal has been to leverage the symmetry of small molecules for quantum computation, by exploiting a coupling between orbital angular momentum and nuclear spin. Such a coupling may prove useful in implementing STP, but in this paper we have considered the abstracted problem where this simplest possible measurement, $s/t$, of a collective spin state is taken as the primitive.

In contrast to the approach presented in the main text, a more geometrical approach for analyzing the computational power of the $s/t$ model is to consider \textit{no-leakage} or \textit{equiangular} sequences. This approach also originates from \cite{freedman2019quantum} where it is shown that the notion of equiangularity is key to many measurement based models that imitate unitary evolution. Indeed, the main result there, involves a model where the measurement based online part of the computation is done through equiangular sequences, as these are sequences in which any undesired outcome of a measurement can be corrected without leakage or starting the computation all over again.

In \cite{freedman2019quantum}, equiangularity of a pair of projection is defined. Here, we need to define equiangularity for an \textit{ordered} pair and generalize some of the results in  \cite{freedman2019quantum}. From now on, a projection $P$ will refer to the operator and the corresponding image (hyper-)plane.

\begin{definition}
We call the \textit{ordered} pair of projection $(Q,P)$ equiangular if $PQP = \alpha^2 P$ for some $\alpha>0$.
\end{definition}

\begin{remark}
\textbf{Equi}angular is a suitable name for the above algebraic equation. In fact, if $PQP = \alpha^2 P$ for some $\alpha>0$, then $\alpha$ is equal to $\cos(\theta)$ with \textbf{all} dihedral angles from $P$ to $Q$ being equal to $\theta$. Also note that $(Q,P)$ being equiangular implies the same for $(1-Q,P)$. The proof of all these involves elementary linear algebra; see \cite{freedman2019quantum}.
\end{remark}

\begin{remark}
An intuitive consequence of the above is that the application of $Q$ to a computational space defined in the plane $P$ is a unitary embedding up to some positive scale, meaning restricted to the image of $P$, the map is a unitary up to some positive scale (see \cite{freedman2019quantum} for the detailed proof).
\end{remark}

The above remark can also be paraphrased as $Q$ defines a \textit{no-leakage} map from the plane $P$ to the plane $Q$. 

Suppose we want to apply the projective measurement $\{Q,1-Q\}$ following $P$. Then, we can always force the outcome to be $Q$. Indeed, assume we get $(1-Q)P$. We try $\{P,1-P\}$ and get either $P(1-Q)P \propto P$ (so back at $P$), or $(1-P)(1-Q)P$. In the latter case, we try $\{Q,1-Q\}$ again, and get either $Q(1-P)(1-Q)P = -QP(1-Q)P \propto QP$ (the desired result), or $(1-Q)(1-P)(1-Q)P = -(1-Q)(1-P)QP = (1-Q)PQP \propto (1-Q)P$, which means we are back at our first unsuccessful measurement. Notice the last case happens with a fixed probability depending on $\alpha$. Thus, repetition always exponentially suppresses the probability of failure.

The above process is the key reason behind the use of equiangular ordered pairs in a measurement model, as it ensures \textbf{recovery} in case of failure. Let us now define an equiangular \textit{sequence} of projections, where the same recovery protocol works.

\begin{definition}
An equiangular sequence $(P_k, \ldots ,P_1)$ satisfies:
\begin{align*}
    P_iP_{i+1}P_i\ldots P_1 = \alpha_{i,i+1}^2P_i\ldots P_1,  \ \ \alpha_{i,i+1} >0 , \ \ \text{for } 1\le i \le k-1.
\end{align*}
\end{definition}
\begin{remark}
In an equiangular sequence, applying $P_{i+1}$ onto the image of $P_i\ldots P_1$ is a unitary embedding up to some positive scale. The argument is simply a more involved version of the ordered pair case. Recovery is similarly ensured, i.e. we can always force $P_{i+1}$ on $P_i\ldots P_1$.
\end{remark}

Next, we define a no-leakage sequence. Intuitively, this is a sequence where the measurements do not ``leak'' quantum information to the environment.  Precisely, we mean:
\begin{definition}
The sequence $(P_k, \ldots ,P_1)$ is no-leakage if the projection $P_{i+1}$ is a unitary embedding up to some scale from the image of $P_i\ldots P_1$ for all $1 \le i \le k-1$ to a subspace of $P_{i+1}$.
\end{definition} 

Is a no-leakage sequence necessarily equiangular? The answer is yes when $k=2$. In general, linear algebra tells us that the orthogonal eigenspaces $E_i$ of $P_iP_{i+1}P_i$ that are a subspace of $P_i$ decompose $P_i = \sum E_i$ to subspaces that have all dihedral angles with $P_{i+1}$ equal. It is not hard to show that in an equiangular sequence, the image of $P_i\ldots P_1$ falls inside one of those eigenspaces and $\alpha_{i,i+1}$ is the cosine of that dihedral angle. However, in a no-leakage sequence, the only requirement is that the image of $P_i\ldots P_1$ is a plane with dihedral angles all equal with $P_{i+1}$, and that plane could be formed by some complicated combination of vectors in different eigenspaces. 

In contrast to the equiangular case, there is generally no guarantee for recovery in a no-leakage sequence. Still, such sequences are of interest in the case of PostSTP, since by definition, we are allowed to force the outcome.

Now, we investigate the computational power of each type of sequence, first equiangular sequences of $s/t$. For a computation to be done on a spin-0 subspace of $2N$ spins, called $spin_0(2N)$, initialized at a singlet dimerization, we use an even number $2M$ of ancillas, similarly initialized at a singlet dimerization. The sequence starts and ends with the projection onto a particular singlet state for the ancilla pairs, i.e. $\prod_{\text{ancilla pairs}} s $, and in between, $s/t$ measurements can be made on the $N+M$ pairs. This ensures that we start and end at the same computational space $spin_0(2N)$. Hence we define:

\begin{definition}
An $s/t$ equiangular sequence refers to an equiangular sequence of the form 
$$(\prod_{\text{ancilla pairs}} s , P_k,\ldots, P_1, \prod_{\text{ancilla pairs}} s ),$$ 
where $P_i$ are $s$ or $t$ on a pair from the $2(N+M)$ spins.
\end{definition}

Equiangularity forbids a dimension decrease of the computational space, as the map at each stage must be unitary. Notice dimensional decrease is even 'worse' than leakage, as leakage can occur if the map is invertible but not unitary. Thus at each stage of the computation, we must observe a matching of $M$ pairs of spins on which a projection have been made, meaning that each projection $P_i$ is applied on a pair that shares exactly one spin with the previous $M$ pairs; sharing no spins would decrease dimension, and sharing both spins would either decrease dimension or be a redundant projection. This is shown in more details in the proof of the next lemma.

For example, let us denote all ancillas by $a_1,\ldots, a_{2M}$ and the computational spins by $c_1,\ldots,c_{2N}$. At first, $a_{2i-1},a_{2i}$ are paired in a singlet state. Then, it is not hard to show that $P_1$ can be any $s$ or $t$ on any pair of the form $a_ic_j$. If $M=1$ (only two ancillas), then each step involves applying $s$ or $t$ on a pair sharing one spin with the previous measurement. Of course, to satisfy equiangularity. there may not be complete freedom in choosing the other spin. If $M>1$, then assuming $P_1$ applies on $a_1c_1$, $P_2$ must share a spin with $\{a_1,c_1,a_3,a_4\}$; notice that $a_2$ has been discarded, as it was paired with $a_1$.
\begin{remark}
We may allow consecutive projections onto pairs that are disjoint, if the goal is to apply an equiangular sequence on a particular \textbf{subspace} of $spin_0(2N)$, i.e. the computational space is a subspace of $spin_0(2N)$. We explored this approach in our simulations to some extent, but without success.
\end{remark}

We now show that no-leakage sequences (and therefore equiangular sequences) with one pair of ancillas only perform signed permutations. This limitation is a reason behind the use of alternative/more analytical methods in \cref{lemxzst}. 

\begin{lemma}
With $M=1$ pair of ancillas, every $s/t$ no-leakage sequence is some signed permutation on the $2N$ computational spins.
\end{lemma}
\begin{proof}
Let us assume $N=2$. Ancillas are $a_1,a_2$ and computational spins are $c_1,...,c_4$. At the start, up to permutation, the only possible measurement is on $a_1,c_1$. Indeed, if we measure $a_1,a_2$ then it leaves the state unchanged and if we measure any $p,q$ with $p,q \in \{c_1,\ldots,c_4\}$, then it decreases the dimension of the computational space.  So we must measure one ancilla and one non-ancilla; without loss of generality, let it be $a_1$ and $c_1$.
 
If the outcome of the $a_1,c_1$ measurement is singlet, then the $a_1,a_2$ ancilla pair has simply been replaced by an $a_1,c_1$ ancilla pair with $c_1$ being teleported to $a_2$. This is a simple permutation. So, we may assume the outcome is triplet.

What is the next possible measurement?  It can not be a pair $p,q\in \{c_2,c_3,c_4\}$, as that would again decrease dimension. One may readily see this as spins $c_2,c_3,c_4$ have not been touched yet, so for some initial states it is possible that such a pair is definitely in a singlet, but not for all.

Notice the state after the first measurement is symmetric under $a_1$ and $c_1$ interchanged, as they are a triplet. So suppose the pair $a_1,a_2$ or $c_1,a_2$ is measured. If the outcome is a singlet, the whole sequence has gained nothing: we return to $a_1,a_2$ being an ancilla in a singlet, and the initial spin $c_1$ is returned back to itself after this sequence. So, suppose we measure $a_1,a_2$ and the outcome is a triplet. 

To analyze this situation, let us make some notation for the state after the first measurement. After the first measurement, the state is $(1+\Sw(c_1,a_1)) \psi$ where $\psi$ is the initial state.  The state $\psi$ is annihilated by the projection onto $a_1,a_2$ triplet. So, we can assume the state after the first measurement is $\Sw(c_1,a_1) \psi$. However, this is equivalent to the initial state up to permutation of spins (i.e. we can replace the first measurement by a permutation) so nothing is gained.

The next possibility for the second measurement is to measure $p,q$ with $p=a_1$ and $q \in \{c_2,c_3,c_4\}$. Equivalently, we could pick $p=c_1$, as $a_1,c_1$ are symmetric. This case also leaks information, though a more detailed check is needed. If spin $c_1$ started out as some $\sigma$ at the start of the sequence, then spin $a_1$ and spin $c_1$ both are “correlated” with that initial $\sigma$ after the first measurement. Therefore, measurement on such a $p,q$ pair leaks information (this can also be checked on a computer).

Finally, we might pick $p=a_2$ and $q\in \{c_2,c_3,c_4\}$ but this also leaks information. Say $q=c_2$ is picked. If at the start of the sequence, we had spins $c_1,c_2$ in a singlet, then a measurement of $a_2,c_2$ will always yield triplet, therefore leakage occurs in this case as well.

Notice the proof does not rely on $N=2$ and works for all $N$.
\end{proof}
\begin{remark}
Given \textbf{equiangular} sequences, computer search, of hundreds of thousands iterations for ten different random seeds for producing random $s/t$ sequences, shows that increasing the number of ancillas does not make any difference in the statement of the theorem above.
\end{remark}

As previously mentioned, there could be interesting operators for PostSTP from no-leakage sequences. In fact, by increasing the number of ancillas, the previous lemma no longer holds:
\begin{lemma}
With no-leakage sequences, we can obtain an infinite order unitary on $2N=4$ spins given $2M=6$ ancillas and a unitary which is not a signed permutation given $2M = 4$ ancillas.
\end{lemma}
\begin{proof}
The proof is by computer search. The following is no-leakage:
\begin{align*}
(s_{a_1,a_2}s_{a_3,a_4}s_{a_5,a_6}, t_{a_6,c_1},s_{a_2,a_6},t_{a_6,c_2},s_{a_4,a_6},t_{a_6,c_4},s_{a_1,a_2}s_{a_3,a_4}s_{a_5,a_6}),
\end{align*}
where $s_{-},t_{-}$ is singlet or triplet projection on pair $-$. The sequence gives a unitary with infinite order on $spin_0(4)$, as its eigenvalues are $\frac{3\sqrt{3}}{2\sqrt{7}} \pm \frac{1}{2\sqrt{7}} i$, which correspond to an irrational angle \cite{mo}.

The following sequence of projections is no-leakage and gives a unitary with order 12 on the space $spin_0(4)$, hence can not be a signed permutation on 4 spins:
\begin{align*}
    (s_{a_1,a_2}s_{a_3,a_4}, t_{a_3,c_3},t_{a_2,c_2},s_{a_2,c_3},t_{a_3,c_2},s_{a_3,c_4},     s_{a_1,a_2}s_{a_3,a_4}).
\end{align*}
\end{proof}

Unlike the case of $N=2$, computer search surprisingly was unable to get any interesting operator for $N>2$, even when using higher number of ancillas and different variations on the operators that can be used (for example, simultaneous projection of triplet and/or singlet on different pairs). Thus, we are unable to use no-leakage sequences to prove PostSTP = PostBQP.

We can relax the requirement even further, by simply demanding the result of the sequence to be no-leakage, i.e. give a unitary operator, thereby not requiring \textit{each} step to be no-leakage. One should imagine an elliptical distortion created by one projection being canceled by a subsequent one. This is still allowed as in PostSTP we can force the outcome. This makes it possible to obtain infinite order unitaries with a sequence of length 3 on $spin_0(4)$ with only 2 ancillas, such as
$$(s_{a_1,a_2},t_{a_2,c_2},t_{a_2,c_1},t_{a_2,c_4},s_{a_1,a_2}).$$
This corresponds to a unitary with eigenvalues $\frac{3\sqrt{3}}{2\sqrt{7}} \pm \frac{1}{2\sqrt{7}} i$, even though $t_{a_2,c_1}$ leaks information (gives an invertible but nonunitary map). However, computer search does not deliver any interesting unitary (i.e. not a signed permutation) given higher numbers of spins and ancillas.
 
\section{Statistics of amplitudes from random $s/t$ sequences}\label{statsofampsappendix}
In this appendix, we provide informal arguments independent from the main text, supplemented with numerical results which support the inference that the $s/t$ model is not classically simulable.  Indeed, we believe that it is not possible to efficiently classically compute the probability of a single measurement being $s$ or $t$ with inverse polynomial error. We present a series of numerical experiments, and then finally use these to propose an average case sampling problem in \cref{STSampleproblem}.

To this end, we examine the statistics of amplitudes from random $s/t$ sequences. Suppose we start with a singlet dimerization on $2N$ spins and pick random pairs and measure whether they are singlet or triplet without postselection. After many rounds of measurement, we pick one final random pair and consider the probability that the pair is a singlet. If this probability is exponentially close to $\frac{1}{4}$, this is not very useful as it can be classically simulated trivially by simply guessing that
the value is $\frac{1}{4}$. On the other hand, if this probability has some fluctuations of magnitude $\frac{1}{poly}$, then it may be hard to simulate classically, and allow one to devise an experiment similar to the Google ``quantum supremacy'' experiment \cite{arute2019quantum}.

However, the naive experiment discussed above fails to deliver any interesting result. To understand this, consider the random sampling performed in \cite{arute2019quantum} using random circuits and individual spin measurements instead of our $s/t$ measurements.  There, the standard approximation is to guess that the random circuit produces a random state (Haar random, as an approximation). This random state is well described by choosing the amplitude of every single computational basis state from a Gaussian distribution.  So, if one measures a single spin on the final state, the distribution is exponentially close to flat ($\frac{1}{2},\frac{1}{2}$ probability). The reason is, for $2N$ spins, there are are $2^{2N-1}$ basis states where the selected spin is up and $2^{2N-1}$ where it is down, and adding all those random numbers gives very small fluctuations that cancel each other out to give a probability exponentially close to ($\frac{1}{2},\frac{1}{2}$). The same holds if one measures $2,3,4, \ldots, O(1)$ spins; their joint probability distribution is still very close to flat.  So, we expect a similar behavior that the simple protocol of the above paragraph will not be too interesting, and we give a more interesting protocol later.

We can see this  flatness in practice using a simulation of a sequence of $1000$ measurements on $18$ spins (\cref{fig:stprobs}) in this protocol, followed by computing the probability of triplet on each pair, of which there are $153$. At first glance, there seem to be interesting fluctuations with high magnitude, but what we actually see, as explained later, is a rather boring deviation from exponential flatness. As expected, most of the graph is hovering around $\frac{3}{4}$; but notice that there some triplets and  singlets (values $1$ and $0$ on the graph respectively). We present an informal argument regarding these (extreme) fluctuations, why they exist even for very long runs like $1000$ or more, and why they are not an indication of quantum supremacy.

\begin{figure}[h]
    \centering
    \includegraphics[scale = 0.5]{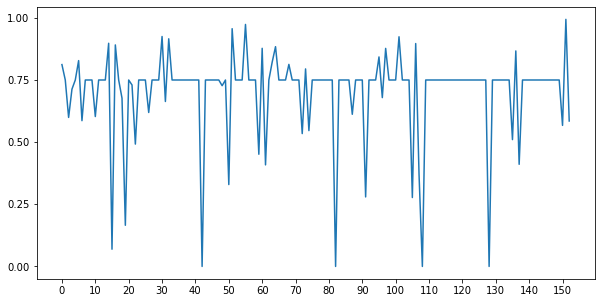}
    \caption{Probability (y-axis) of being a triplet pair for the $\binom{18}{2}=153$ pairs (x-axis) of 18 spins after 1000 random measurements performed on one randomly chosen sequence of 1000 pairs from a singlet dimerization.}
    \label{fig:stprobs}
\end{figure}

Consider a gas of spins and draw a straight line between two spins when they form a singlet, and a wiggly line for a triplet, and otherwise no line. Suppose the gas reaches a state where there are no lines. Then, as we select a pair and measure $s/t$, some line (straight or wiggly) forms. By selecting another random pair, likely disjoint from the first, another line forms. By the time about $(1-1/\sqrt{2})2N$ of the spins are paired off by lines of one kind or the other, it becomes equally probable that the next randomly selected pair will touch/not touch a paired spin and that the measurement will break the preexisting line and replace it with another. We should note that breaking a line is the worst case, corresponding to the singlet measurement outcome, and we are assuming such an outcome for every measurement, giving us the ratio $(1-1/\sqrt{2})$ (if triplet is the outcome, the previous line may remain, and we get more lines as a result). On this ``fresh'' subpopulation, the process, even if it is at step 1001, looks like it is starting from scratch. More generally, there is also an $O(1)$ fraction subpopulation whose parents, or grandparents were all fresh - part of a singlet sub-dimerization, which could provide an explanation on the rest of the fluctuations in \cref{fig:stprobs}. This constant ``refreshment'' prevents exponential flatness, but should not be mistaken for evidence of quantum supremacy. We noted how sampling an $O(1)$ number of spins would not break the exponential flatness (other than the deviations described in the above paragraph). However, once we sample almost all the spins, then amplitudes can get interesting. For example, after sampling $2N-k$ spins, for some $k$, then the deviation of the next spin from flatness is only exponentially small in $k$. So, the last few measurements are not close to flat.

To implement such a sampling, imagine a protocol that ``detangles'' a state: take a pair of spins at random, and measure. If the answer is a singlet, then remove that pair from the collection of spin, and continue the algorithm on the remaining spins.  If it is not a singlet, then take another random pair (potentially not disjoint from the previous pair found to be a triplet). Eventually, this fully detangles the state, and at some point only a few spins are left entangled, and finally none at all.

As one detangles, the probabilities start exponentially close to flat but become less and less flat. When there are only $4$ spins left, the last measurements are far from flat. Notice if we stop at $2$ spins, the result is already determined: a singlet, as the total spin of our original state had to be zero.

Let us call those spins $1,2,3,4$. Then, we might guess that after many measurements, we have a ``random'' spin 0 state on $1,2,3,4$. So, we want to compute: given a random singlet on $4$ spins, what is the distribution of the probability that a given pair, say $1,2$, will be in a singlet?

The space $spin_0(4)$ is two-dimensional, and is spanned by an orthogonal basis which amounts to symmetrizing or anti-symmetrizing on the pair $1,2$, implying there is one state $\ket{S}$ which is a singlet on $1,2$, and one state $\ket{T}$ which is a triplet on $1,2$. We ask: For a randomly chosen state, what is the distribution of probability on the $\ket{S}$ outcome of the measurement?

A simple guess would be a random (real) state of the form $\cos(\theta)\ket{S} + \sin(\theta) \ket{T}$, with $\theta$ picked uniformly on the circle. Since $d \theta = -\frac{d \cos(\theta)}{\sin(\theta)}$, a uniform measure on $\theta$ is a measure proportional to $1/|\sin(\theta)|$ on the amplitude $\cos(\theta)$ (\cref{fig:costheta}) and as we would measure probabilities, we should consider $\cos(\theta)^2$. However, due to the subpopulation refreshment phenomena described earlier, we would expect to see certain number of spikes on the distribution over probabilities $1,\frac{3}{4},\frac{1}{4},0$ and possibly other numbers. To sum it up, we expect a continuous measure which hopefully delivers arbitrary looking numbers, plus some inevitable discrete terms.

\begin{figure}[h]
    \centering
    \includegraphics[scale = 0.65]{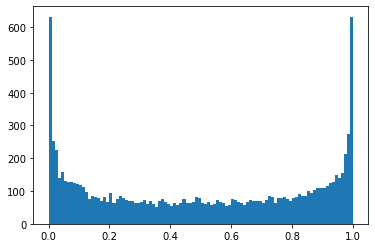}
    \caption{Histogram distribution of $\cos(\theta)^2$ for $\theta$ picked uniformly from the circle using 10000 samples and 100 equal bins.}
    \label{fig:costheta}
\end{figure}

We illustrate the probability of $\ket{S}$ for a simulation of 10 runs of the algorithm on 1000 random sequences of 1000 measurements on 14, 16, and 18 spins (\cref{fig:STprobs}). We see some spikes at $1,\frac{3}{4},0$ for $\ket{S}$ and $1,\frac{1}{4},0$ for $\ket{T}$. Overall, we see a continuous measure plus some discrete terms, and the continuous measure delivers the arbitrary looking numbers, apparently quite distinct from $\cos(\theta)^2$. This may indicate a difficulty for classical simulation, and a route to quantum supremacy.

\begin{figure}[h]
    \centering
    \includegraphics[scale = 0.5]{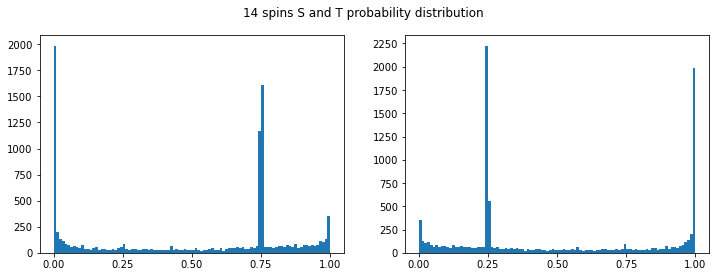}
    \includegraphics[scale = 0.5]{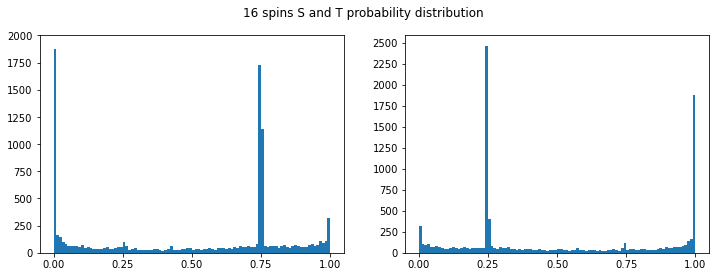}
    \includegraphics[scale = 0.5]{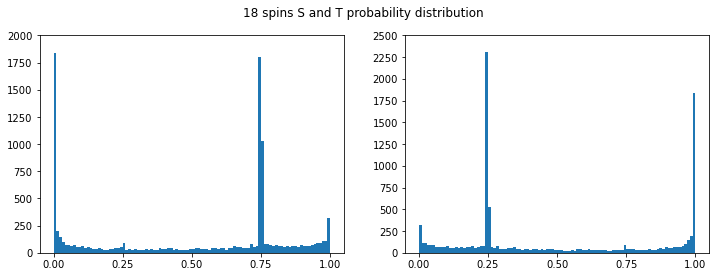}
    \caption{Left/Right histogram shows the distribution of the probability of $\ket{S}$/$\ket{T}$ being the outcome of the measurement on the last 4 spins after implementing the detangling process. For each value of $2N \in \{14,16,18\}$, the 10000 datapoints are from 10 runs of the algorithm on 1000 random sequences of 1000 measurements.}
    \label{fig:STprobs}
\end{figure}

Thus, we propose the following random sampling problem, \textit{\textbf{STSample}}, which can be viewed as an average case sampling problem:
\begin{problem}\label{STSampleproblem}[\textbf{STSample}]
Considering preparing $2N$ spins in a singlet.  Perform $\Theta(poly(N))$ rounds of $s/t$ measurements on random pairs.  Then, apply the detangling protocol, until $4$ spins are left, and then finally perform one last $s/t$ measurement.  Given a choice of random pairs, this defines a probability distribution over bit strings (where each bit being $0$ or $1$ denotes whether a measurement was $s$ or $t$).
Suppose the pairs are chosen randomly and then a bit string is chosen randomly from this distribution, and a classical computer is given as input the list of pairs and given all but the last bit of the bit string.  Can it, with probability $1-o(1)$, efficiently compute the probability distribution of the last bit with at most inverse polynomial error?
\end{problem}

\section{A numerical/representation theory approach to PostSTP = PostBQP}
We present another alternative approach, this time involving representation theory \footnote{We thank Monica Vazirani for her tutoring on the representation theory of the symmetric group}, to proving PostSTP = PostBQP. The ``theory'' is easy, but numerical difficulties have prevented us from making this approach rigorous.

Given $2N=8$ spins in a spin-0 subspace and $2M$ spins as ancillas initialized in a singlet dimerization, taking out dilation, we aim to prove density in the group $SL(14,\mathbb{R})$, where $14 = \dim spin_0(8)$. Then we define our logical qubit as the two dimensional $spin_0(4)$. We may think $spin_0(8)$ as containing two $spin_0(4)$ qubits: $spin_0(1,2,3,4)$ ad $spin_0(5,6,7,8)$, and 10 additional (unwanted) directions. This density result would give the entangling gates for the logical qubits, as the $SO(4)$ gates are included in the larger $SL(14,\mathbb{R})$ group.

We know that by simply using singlet projections, we can permute the spins. Thus, we have an action of the symmetric group $S_8$ is given on $spin_0(8)$. This \textit{irreducible} representation corresponds to the Young tableau with partition $(4,4)$, as also shown in \cite{jordan2010permutational}. An action is induced (by conjugation) on the endomorphism algebra $\text{End}(spin_0(8)) \cong \text{M}(14,\mathbb{R})$, and similarly on its Lie algebra $\mathcal{A}$ which is isomorphic to $\text{M}(14,\mathbb{R})$. Note that we mention the Lie algebra as we will take logs of operators later on.

Using the mathematical package SAGE, $\mathcal{A}$ decomposes, in Young diagram notation, to $(8) \oplus  (6,2) \oplus (5,1,1,1) \oplus (4,4)  \oplus (4,2,2) \oplus (3,3,1,1)\oplus  (2,2,2,2)$ under the aforementioned action. Notice $(8)$ is the trivial irrep containing the identity. Henceforth, we replace $O$ by normalized $\frac{O}{\det(O)}$. This projects out the $(8)$ summand leaving $\mathcal{A}_0 \cong \mathfrak{sl}(14,\mathbb{R})$.

We thus consider only these normalized operators with determinant 1 in $\text{SL}(14,\mathbb{R})$ given as sequences of $s/t$. To get a dense set, ideally the set of gates should be closed under inverse, allowing a variant of the Solovay-Kitaev algorithm \cite{kitaev1997quantum} to be implemented. But how can we approximate the inverse of a sequence of $s/t$? If the operator is not orthogonal, we can not simply take the conjugate sequence, which is the sequence applied in reverse order.

Take the (smallest) log of one of these operators (assuming it exists), say $\log(O)$. Then consider the orbits of $\log(O)$, which are also the logs of the orbit of $O$. Adding them up produces an element $X = \sum_{\rho \in S_8} \rho(\log(O))$ in $\mathcal{A}_0$ which is invariant under $S_8$. Since $(8)$ has been projected out, $0$ is the only fixed vector in $\mathcal{A}_0$. Thus we have,
$$X = 0 \implies -\log(O) = \sum_{id \neq \rho \in S_8} \rho(\log(O)).$$
Exponentiation and using Baker-Campbell-Hausdorff formula for the error analysis now allow us to approximate the inverse of $O$.

In the numerical studies, we have normalized the operators by the determinant, thus eliminating the trivial irrep $(8)$, leaving six irreps spanning the Lie algebra $\mathcal{A}_0 \cong \mathfrak{sl}(14,\mathbb{R})$. To achieve density in $SL(14,\mathbb{R})$, we need to search for:
\begin{lemma}
Density in $\text{SL}(14,\mathbb{R})$ follows given a set of operators ``arbitrarily'' close to identity (e.g. in the operator norm) with logs having nonzero component in each of the six irreps.
\end{lemma}
\begin{proof}
Given such an operator $O$, define $\mathcal{V}_O := \text{span}\left<\{\rho(\log(O))\}_{\rho \in S_8}\right>$. Clearly $\mathcal{V}_O$ is invariant and since it has nonzero projection onto each of the six irreps, it must be the whole Lie algebra $\mathfrak{sl}(14,\mathbb{R})$.

However, density must be achieved by \textit{positive} integral linear combinations of $\{\rho(\log(O))\}_{\rho \in S_8}$, as we only have access to positive integer multiplies of $\rho(\log(O))$ by composing their corresponding $s/t$ sequences. But we previously showed how to approximate inverses (negative multiples) with positive linear combinations, thus giving access to \textit{integral} linear combinations. Having operators arbitrarily close to $id$, meaning (the smallest) $\log(O)$ very close to zero (in, say, the operator norm), serves the purpose of establishing density with just positive integral linear combinations.
\end{proof}
Moreover, it is not necessary to have operators \textit{arbitrarily} close to identity, as the threshold theorem states; see also the last argument in \cref{postbqpsubsec} regarding the issue of error. Nevertheless, the threshold error required in this case is very small, especially as we build the inverse of an operator by taking $8!-1$ other operators, which introduces even more errors. A very rough estimate would be $10^{-10}$, but it is likely much less than that, as the usual error correction threshold is around $10^{-3}$ and $8!-1 \sim 10^{4}$ terms are used in a composition to compute the inverse of any element. Thus, numerical assisted proof by searching sequences of $s/t$ is out of reach. However, a search can provide evidence of the existence of such operators, by showing one can find operators closer and closer to identity as the length of the $s/t$ sequence increases. We have found such examples in the smaller case of 4 computational spins.

We find two operators 
$$O_1 = \begin{pmatrix} \frac{\sqrt{15}}{2\sqrt{7}} & \frac{1}{2\sqrt{35}} \\ -\frac{\sqrt{5}}{2\sqrt{7}} & \frac{9\sqrt{3}}{2\sqrt{35}} \end{pmatrix}, \ \ \ O_2 = \begin{pmatrix} 0.7 & -\frac{\sqrt{3}}{10} \\ -\frac{1}{2\sqrt{3}} & 1.5 \end{pmatrix},$$
corresponding to
\begin{align*}
   O_1 \leftarrow (s_{a_1,a_2},t_{a_2,c_2},t_{c_1,c_2},t_{a_1,c_1},t_{a_1,c_4}, t_{c_3,c_4}, t_{c_2,c_3}, t_{c_1,c_2}, t_{a_1,c_1},s_{a_1,a_2}), \\
O_2 \leftarrow (s_{a_1,a_2}, t_{a_2,c_2},t_{c_1,c_2},t_{c_1,c_3},t_{a_1,c_3},t_{a_1,a_2},t_{a_2,c_2}, t_{c_1,c_2},t_{a_1,c_1}  ,s_{a_1,a_2}),
\end{align*}
both satisfying $\prod_{\rho \in S_4} \rho(O_i) = id$ where $\rho$ acts by conjugation. Therefore, we have exact access to the inverse of these operators and can compute their commutators. $O_1$ and $O_2$ are obtained by a sequence of length $8$ of triplet measurements on $4$ computational spins and $2$ ancillas. It should be noted that these operators are not orthogonal, and their orbits $\{\rho(O_i)\}_{\rho \in S_4}$ do not pairwise commute. Hence the identity $\prod_{\rho \in S_4} \rho(O_i) = id$ holds for \textit{non}obvious reasons. Following the technique in Solovay-Kitaev theorem \cite{kitaev1997quantum}, by taking repeatedly the commutators $[O_1,[O_1, \ldots,[O_1,O_2]]\ldots ]$, we can find operators closer to identity (\cref{fig:commtoid}).
\begin{figure}[h]
    \centering
    \includegraphics[scale=0.4125]{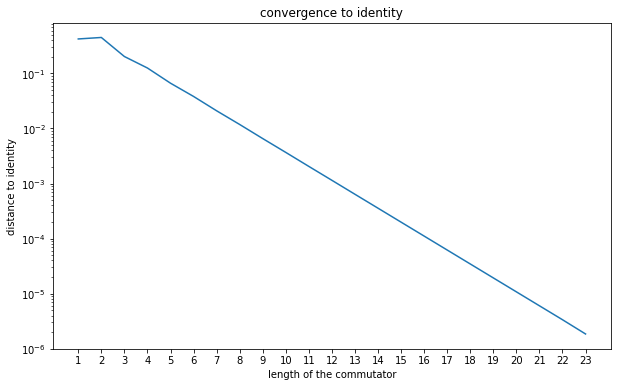}
    \caption{Distance of $[O_1,[O_1, \ldots,[O_1,O_2]]\ldots ]$ ($M$ times) to identity. The x-axis is the length of the commutator which is $M$ and the y-axis has been log-scaled to show the exponential decay of distance to identity.}
    \label{fig:commtoid}
\end{figure}

\end{document}